\documentclass[11pt]{article}

\newcommand{\blind}{0}

\usepackage[english]{babel}
\usepackage[T1]{fontenc}
\usepackage{amsthm,amsmath,amssymb,mathrsfs,dsfont,mathtools, amsfonts}
\usepackage[colorlinks = true,
            linkcolor = black,
            urlcolor  = black,
            citecolor = black,
            anchorcolor = black]{hyperref}
\usepackage[noabbrev, capitalise]{cleveref}
\usepackage{framed}
\usepackage{bbm}
\usepackage{bm}
\usepackage{graphicx}
\usepackage{xcolor}
\usepackage{booktabs}
\usepackage[ruled]{algorithm2e}
\usepackage{algorithmic, setspace}
\usepackage[shortlabels]{enumitem}
\usepackage{verbatim}
\usepackage{float}
\usepackage{nccmath}

\usepackage[authoryear, round]{natbib}
\newtheorem{theorem}{Theorem}

\newtheorem{lemma}{Lemma}
\newtheorem{proposition}{Proposition}
\theoremstyle{definition}

\usepackage{subcaption}

\newcommand{\UU}[1]{U^{(#1)}}
\newcommand{\uu}[2][{}]{u^{(#1)}_{#2}}
\newcommand{\sk}[2][{}]{S^{(#1)}_{#2}}
\newcommand{\R}{\mathbb{R}}
\newcommand{\N}{\mathbb{N}^+}
\newcommand{\G}{\mathcal{G}}
\newcommand{\T}{\mathcal{T}}
\newcommand{\Z}{\mathcal{Z}}
\newcommand{\Y}{\mathcal{Y}}
\newcommand{\Q}{\mathcal{Q}}
\newcommand{\tvec}{\text{vec}}
\newcommand{\tdiag}{\text{diag}}

\newcommand{\pr}{\mathrm{pr}}
\newcommand{\cov}{\mathrm{cov}}

\begin{document}

\def\spacingset#1{\renewcommand{\baselinestretch}%
{#1}\small\normalsize} \spacingset{1}


\if0\blind
{
  \title{\bf Bayesian Adaptive Tucker Decompositions for Tensor Factorization}
  \author{Federica Stolf\\ 
    Department of Statistical Science, Duke University, Durham, NC, USA\\
    and \\
    Antonio Canale \\
    Department of Statistical Sciences, University of Padova, Padova, Italy}
  \maketitle
} \fi

\if1\blind
{
  \bigskip
  \bigskip
  \bigskip
  \begin{center}
    {\LARGE\bf Title}
\end{center}
  \medskip
} \fi

\bigskip
\begin{abstract}
Tucker tensor decomposition offers a more effective representation for multiway data compared to the widely used PARAFAC model. However, its flexibility brings the challenge of selecting the appropriate latent multi-rank. To overcome the issue of pre-selecting the latent multi-rank, we introduce a Bayesian adaptive Tucker decomposition model that infers the multi-rank automatically via an infinite increasing shrinkage prior. The model introduces local sparsity in the core tensor, inducing rich and at the same time parsimonious dependency structures.
Posterior inference proceeds via an efficient adaptive Gibbs sampler, supporting both continuous and binary data  and allowing for straightforward missing data imputation when dealing with incomplete multiway data.
We discuss fundamental properties of the proposed modeling framework, providing theoretical justification. Simulation studies and applications to chemometrics and complex ecological data offer compelling evidence of its advantages over existing tensor factorization methods.

\end{abstract}

\noindent%
\vfill

\newpage
\spacingset{1.2} 

\section{Introduction}
In many applications, data naturally come in the form of multiway arrays or tensors such as recommender system \citep{Bietal}, ecology \citep{Frelat3D} and chemometrics \citep{tomasi2005parafac}.
 In a general setting, let  $\Y = \{y_{i_1 \cdots i_K}, \ i_k = 1, \dots, n_k\}$ denote a $K$-way tensor and we assume that $\mathcal{Y} \in \mathbb{R}^{n_1 \times \dots \times n_k}$, with the integer $n_k$ being the dimension of $\Y$ along the $k$-th way. We are motivated by settings where $K$ is relatively small. 
 For instance, in ecological studies examining species co-occurrence, $K=3$ with data organized in a three-way array structured as site $\times$ species $\times$ year.
One naive approach to analyze this type of data is to flatten the data into a matrix and then use existing methods for matrix factorization. However, this transformation fails to preserve the data structure, making learning low-dimensional relationships more difficult. 

Tensor decomposition methods based on the low-rank tensor approximation \citep{KoldaandBader} offer a powerful tool to analyze multiway data, by providing a concise representation that effectively captures the dominant underlying features. 
The two main forms of tensor decomposition,  which extend the matrix singular value decomposition, are  the Tucker decomposition and its restriction known as the parallel factor (PARAFAC) decomposition. PARAFAC factorization decomposes a tensor into a sum of $R$ rank-one tensor, where $R$ is referred to as the rank of the tensor, while Tucker factorization decomposes a tensor into a small core tensor multiplied by a matrix along each mode. Tucker representation is more flexible than the PARAFAC, as it includes a multi-rank $(R_1, \dots, R_k)$ that allows for different decomposition ranks for various tensor dimension margins. This flexibility enables a more effective data compression, as not all the dimensions of a tensor might require all $R$ principal components in order to be faithfully represented by such a decomposition. In this case, PARAFAC rank can be considerably higher than the $K$ ranks composing the Tucker multi-rank. 
Tucker representation is particularly effective when the tensor dimensions $n_1, \dots, n_K$ differ from each other.

Bayesian tensor decomposition models \citep{prob_tensor}  are naturally appealing, as they allow for straightforward quantification of uncertainty in estimating model parameters, while providing a principled mechanism for dealing with missing data. 
Partially observed multiway data, indeed, are abundant in numerous fields; see for example  \cite{tomasi2005parafac} for applications in chemometrics.
Missing data can arise in a variety of settings due to loss of information, errors in the data collection process, or costly experiments. 
The problem of recovering latent structures from incomplete tensors to reconstruct the missing entries is often referred to as tensor completion.
For tensor factorization models in a Bayesian framework see for instance \citet{HoffSeparable} and \citet{xu2013bayesian}. 
Refer to \cite{shi2023bayesian} for a general review of Bayesian methods in tensor analysis, including topics beyond tensor completion, such as tensor regression models \citep{Guhaniyogietal2017, papadogeorgou2021soft}. 

An open problem in low-rank tensor factorization literature is the choice of the latent rank.
Current practice typically specifies the ranks in advance or uses heuristics for choosing the ranks based on repeated fitting under different choices. However, in the Tucker case, an exhaustive search over the space of possible multi-ranks may be infeasible, since the number of  possibilities increases exponentially with the order of the tensor. 
In principle, one can apply a matrix rank determination method for matrix separately to each mode-$k$ matricization of a tensor, such as the criteria proposed in \cite{ahn2013eigenvalue}. However, these methods are inadequate for binary tensor data or tensors with missing values.
A more elegant approach relies on Bayesian nonparametric methods with \citet{rai2014scalable} and \citet{ZhaoCP2015} being interesting contributions in the PARAFAC context. 

Motivated by the above considerations and the need for an intermediate formulation that balances the simplicity of the PARAFAC with the flexibility of the Tucker in a parsimonious way, we introduce a flexible Bayesian Tucker decomposition model where the multi-rank is not pre-specified but is adaptively learned from the data. We rely on successful Bayesian nonparametric models that include more than enough latent elements, with their selection induced by appropriate shrinkage priors that adaptively eliminate unnecessary components.
Exploiting the literature on infinite factor model \citep{Bhattacharyadunson11, CUSP, schiavon2022generalized} we introduce a multiway extension of the cumulative shrinkage process (CUSP) of \citet{CUSP} to apply joint shrinkage for multiple latent quantities along different dimensions.
While the proposed multiway CUSP builds on the foundation of \citet{CUSP},  it addresses unique challenges inherent to tensor-structured data. For example, besides providing shrinkage towards low-rank decomposition of the tensor margins, the proposed method also enables local shrinkage in the latent factors by assuming a suitable shrinkage prior for the elements of the core tensor.
 Incorporating zero entries in the core tensor is crucial for reducing complexity, as inference in the Tucker model scales with the size of the core tensor.
 By shrinking towards zero the elements of the core tensor we are moving it to something similar to a superdiagonal tensor  with different dimension margins, and thus bridging PARAFAC and Tucker decompositions. For a contribution in this direction for probability tensors for multivariate categorical data see the collapsed Tucker decomposition in \cite{Johndrowetal2017}.

In the next sections, we first establish the notation and then present our proposed Bayesian adaptive Tucker (AdaTuck) decomposition.  This novel Bayesian approach estimates the latent multi-rank of tensors using an increasing shrinkage prior while enforcing local sparsity.
We will discuss appealing theoretical properties and outline straightforward implementations for modeling both continuous and binary incomplete multiway data using fully conjugate Bayesian analysis.  Posterior distributions are approximated through closed-form Gibbs sampling updates, making inference highly scalable and straightforward to predict missing entries in the data, as discussed in Section \ref{sec:computation}. The empirical performances of the proposed approach are tested against alternatives in Section \ref{sec:illustration}, including simulated and real data examples from a chemometrics application and a complex ecological study. 
Section \ref{sec:end} presents the main findings of the paper and outlines possible directions for future research. Technical details, including proofs and additional numerical results are reported in the Supplementary Materials.

\section{Adaptive Tucker factorization}

\subsection{Tucker tensor decomposition}
Suppose we observe $\Y = \Z + \mathcal{E}$, where $\Z$ is a mean tensor with $\Z \in \mathbb{R}^{n_1 \times \dots \times n_K}$ describing a signal of interest and $\mathcal{E}$ is a residual noise. We focus on modelling $\Z$ with a Tucker decomposition, i.e. 

\begin{equation} \label{eq:tuck_intro}
    \Z = \sum_{r_1=1}^{R_1} \dots \sum_{r_K=1}^{R_K} g_{r_1 \dots r_K} \uu[1]{r_1} \circ  \cdots  \circ \uu[K]{r_K},
\end{equation}
where $\uu[k]{r_k} =(\uu[k]{1r_k}, \dots, \uu[k]{n_kr_k})^{\top}$, $1 \le k \le K$, and $\G =\{g_{r_1 \dots r_K}, \ r_k = 1, \dots, R_K\}$ is referred to as the core tensor. 
The symbol \lq $\circ$' represents the vector outer product.
Let $\UU{k} = [\uu[k]{r_1},  \dots, \uu[k]{r_K}]$, with $k = 1, \dots, K$, denotes the factor matrix of the $k$-th mode of the tensor. With the Tucker decomposition as given in \eqref{eq:tuck_intro} the tensor element $z_{\mathbf{i}}$, with $\mathbf{i}=[i_1 \dots i_K]$ its $K$-dimensional index vector, can be concisely represented by
\begin{equation} \label{eq:Tuck_elmw}
   z_{\mathbf{i}} =  \sum_{r_1=1}^{R_1} \dots \sum_{r_K=1}^{R_K} g_{r_1 \dots r_K} \prod_{k=1}^K \uu[k]{i_k r_k}.
\end{equation}

In a low-rank framework approach we assume that the elements composing the multi-rank $R_1,  \dots, R_K$ are much smaller than $n_1, \dots, n_K$, yielding a parsimonious representation of $\Z$.
Denoting with vec($\Z$) the result vector of size $n = \prod_{i=1}^K n_i$ from the vectorization of $\Z$, the above Tucker decomposition can be written as
\begin{equation*}
    \text{vec}(\Z) = (\UU{K} \otimes \cdots \otimes \UU{1} ) \text{vec}(\G),
\end{equation*}
where \lq $\otimes$' denotes the Kronecker product and vec$(\T)$ is the vectorization of the tensor $\T$.
If one considers $\UU{k}$ as factor loadings matrix and $g_{r_1 \dots r_K}$ to be the corresponding coefficients, then the Tucker decomposition may be thought of as a multiway analog to factor models.

The PARAFAC decomposition is a constrained version of the Tucker, which restricts the core tensor $\G$ to be superdiagonal with $R_1 = \dots = R_K = R$, factorizing the tensor elements $z_{\mathbf{i}}$ as
\begin{equation*} 
       z_{\mathbf{i}} =  \sum_{r=1}^{R} \lambda_{r} \prod_{k=1}^K \uu[k]{i_k r},
\end{equation*}
with $\lambda_r \ge 0$. This structure imposes that the $r$-th component in one factor matrix can only interact with the $r$-th components in the other factor matrices.

\subsection{Bayesian model specification} \label{sec:modeldescr}
In this section we introduce the proposed Bayesian adaptive Tucker factorization model. 
We focus on a normal likelihood for the noise $\mathcal{E}$, and specifically vec$(\mathcal{E}) \sim N_n(0, \sigma^2 I_n)$, but we will generalize it to the case of binary responses in Section \ref{sec:binProbit}. Therefore, the model likelihood is $\pr(\Y \mid \Z) = \prod_{\mathbf{i} \in I} N(y_{\mathbf{i}} ; z_{\mathbf{i}}, \sigma^2)$, where $I$ is the index set of all observations. Section \ref{sec:binProbit} discusses an extension for binary tensor data.   Marginalizing out the core tensor $\G$, the distribution over the entries of $\Y$ is still Gaussian and can be written as 
\begin{equation}
\label{eq:vecmodel}
     \tvec(\Y) \mid W \sim N_n(0, W \Sigma_g W^{\top} + \sigma^2I_n), 
\end{equation}
where $W = \UU{K} \otimes \cdots \otimes \UU{1}$ and $\Sigma_g = \cov\{\tvec(\G)\} $.

To adaptively learn the latent multi-rank of $\Z$, we propose a new class of priors that applies joint shrinkage of the factors along each mode of the tensor. Our construction builds on the CUSP \citep{CUSP}, originally proposed in the context of factor model for matrix data and successfully extended in different contexts \citep{fruhwirth2023generalized, kowal2023semiparametric, anceschi2024bayesianjointadditivefactor}. Crucially, our contribution differs from these generalizations by specifically targeting multiway data. 
For $k = 1, \dots, K,$ $r_k \in \N,$ and $i_k = 1,\dots, n_k$ we assume for the entries of the factor matrices $\uu[k]{i_kr_k} \sim N(0, \theta^{(k)}_{r_k})$  with
\begin{align} \label{eq:nu}
    &\theta_{r_k}^{(k)} \mid \pi_{r_k}^{(k)} \sim (1-\pi_{r_k}^{(k)}) \text{InvGa}(a_{\theta}, b_{\theta}) + \pi_{r_k}^{(k)} \delta_{\theta_{\infty}}, \\
    &  \pi_{r_k}^{(k)} = \sum_{l=1}^{r_k} \omega_{l}^{(k)}, \quad \omega_{l}^{(k)} = V_{l}^{(k)} \sum_{m=1}^{l-1} (1- V_{m}^{(k)}),  \quad
    V_{l}^{(k)} \stackrel{iid}{\sim }\mathrm{Beta}(1, \alpha^{(k)}),
    \nonumber
\end{align}
where $\mathrm{InvGa}(a,b)$ is the inverse gamma distribution with shape $a$ and scale $b$ parameters, $\delta_x$ is the Dirac delta mass at point $x$. 
Within this specification, the countable sequence of parameters $\theta^{(k)} = \{\theta^{(k)}_{r_k}:r_k \in \N \}$   controls the shrinkage for the elements in the $k$-th mode of the tensor. Employing a stick-breaking construction of the Dirichlet process for the $ \pi_{r_k}^{(k)}$ ensures that the probability assigned to the spike $\theta_{\infty}$ increases with the model dimension $r_k$ and that $\lim_{r_k \to \infty} \pi_{r_k}^{(k)} = 1$ for each $k=1,\dots,K$.
This multiway CUSP prior will  jointly shrink $\uu[k]{i_k r_k}$ towards zero as $r_k$ increases for all tensor dimensions. To allow effective shrinkage of redundant factor matrices elements we set $\theta_{\infty}$ close to zero. The following proposition characterizes the expected value \emph{a priori} of the latent multi-rank under the AdaTuck model.

\begin{proposition} \label{prop:expR}
    Define $c_{r}^{(k)} \mid \pi_{r}^{(k)} \sim Ber(1- \pi_{r}^{(k)})$ and 
    $\theta_{r}^{(k)} \mid c_{r}^{(k)} \sim c_{r}^{(k)} \mathrm{InvGa}(a_{\theta}, b_{\theta}) + (1-c_{r}^{(k)})\delta_{\theta_{\infty}}$    
    with $\pi_{r}^{(k)}$ as defined in \eqref{eq:nu} and let $R^*_k = \sum_{r=1}^{\infty} E(c^{(k)})$ denote the number of active elements of margin k for $k=1,\dots,K$. Then $E\{(R^*_1, \dots, R^*_K)\} = (\alpha^{(1)}, \dots, \alpha^{(K)})$.
\end{proposition}

Proposition \ref{prop:expR} provides a simple and intuitive formulation for the expected value of the latent multi-rank, yielding useful insights for the choice of the hyperparameters $\alpha^{(k)}$. Thus, $\alpha^{(k)}$ could be set to the expected numbers of active components of the $k$-th mode for $k=1,\dots,K$.

This prior involves infinitely many columns for each factor matrix $U^{(k)}$, implying an infinite dimensional core-tensor. However, for practical reasons $\theta^{(k)}$ in \eqref{eq:nu} can be restricted  to finitely many terms, i.e. $\theta^{(k)} = (\theta^{(k)}_1, \dots, \theta^{(k)}_{R_k})$, by letting $V^{(k)}_{R_k} = 1$, resulting in a core tensor of dimension $R_1 \times \dots \times R_K$.  
Truncation is standard practice in nonparametric Bayesian models involving infinitely-many components, ranging from mixture models \citep{ishwaran2001gibbs} to  infinite factor models \citep{Bhattacharyadunson11}.  The following theorem characterizes the bound on the approximation error for the truncated representation of the Tucker factorization, offering theoretical support for this choice.

\begin{theorem} \label{thm:truncation}
    Let $M_{\textbf{i}}^{R_1,\dots,R_K} = \sum_{\mathbf{r}, r_k = R_k+1}^{\infty} g_{\mathbf{r}} \prod_{k=1}^K \uu[k]{i_k r_k} $ be the residuals for a generic index vector $\mathbf{i}$. Assume 
    $E(g_{\mathbf{r}}^2) \leq \zeta$ with $\zeta$ finite.
    Then, $\forall \epsilon>0$,   
    $\pr\{(M_{\textbf{i}}^{R_1, \dots, R_K})^2 > \epsilon\} < (\zeta/\epsilon)  \prod_{k=1}^K \alpha^{(k)}\{\alpha^{(k)}(1+\alpha^{(k)})^{-1}\}^{R_k} $
    for each $i \in I$.
\end{theorem}
Theorem \ref{thm:truncation}  shows that the approximation error for the truncated representation decreases exponentially fast as the truncation level tends to infinity. This suggests that the infinite-dimensional prior in \eqref{eq:nu} is accurately approximated by a conservative truncation. The assumption on the finite second moment for the elements of the core tensor is met with the prior we will use for the elements of $\G$, as defined in the following of this section.

 Often it is also important  to shrink the elements of the core tensor towards zero to reduce model complexity and achieve accurate but regularized estimates of the elements of $\G$.
Introducing local sparsity in the core tensor entails that only summands within the Tucker decomposition that explain additional variance are included in the model. 
To induce local shrinkage towards the elements of the core tensor, we assume that the elements of $\G$ follow a generalized double-Pareto prior \citep{armagan2013generalized} 
\begin{equation} \label{eq:corePar}
    g_{\textbf{r}} \sim N(0, \tau\nu_{\mathbf{r}}), \quad \tau \sim \mbox{Ga}(a_{\tau}, b_{\tau}), \quad  \nu_{\mathbf{r}} \sim \mbox{Exp}(\rho_\mathbf{r}^2/2), \quad  \rho_\mathbf{r} \sim \mbox{Ga}(a_{\rho}, b_{\rho}),
\end{equation}
where $\mbox{Ga}(a,b)$ is a gamma distribution with mean $a/b$ and variance $a/b^2$,  with $\mathbf{r} = [r_1 \dots r_K]$ a $K$-dimensional index vector. 
This is equivalent to assume $ \tvec(\G) \sim N_R(0, \Psi)$, with $\Psi = \text{diag}(\psi_{1\cdots1}, \dots, \psi_{R_1 \cdots R_K})$, $\psi_{r_1 \cdots r_K} = \tau\nu_{r_1 \cdots r_K}$ and $R = \prod_{i=1}^K R_i$.
Integrating over the element-specific scale parameters reduces the prior on $g_{\textbf{r}}$ to a double-exponential (Laplace) distribution centered at $0$ with scale parameter $\rho_\mathbf{r}/\sqrt{\tau}$,  which has heavier tails than the Gaussian distribution while also allocating higher densities around
zero. This produces an analog to the adaptive LASSO as discussed by \cite{armagan2013generalized}.

The proposed prior structure has positive support on both the Tucker and the PARAFAC decomposition, thus representing a bridge between the two models. 
To better understand the implications of this choice, we can look at the covariance between the elements of $\Z$ that arises from the proposed model. The covariance between any two elements $z_{\mathbf{i}}$ and $z_{\mathbf{j}}$ conditioned on the factor matrices is
\begin{equation} \label{eq:corr}
    \text{Cov}(z_{\mathbf{i}},z_{\mathbf{j}} \mid \{U^{(k)}\}_{k=1}^K) = \sum_{r_1=1}^{R_1} \cdots  \sum_{r_K}^{R_K} \psi_{r_1 \cdots r_K} \prod_{k=1}^K \uu[k]{i_k r_k} \uu[k]{j_k r_k}.
\end{equation}
The local sparsity allows for the possibility in \eqref{eq:corr} that some factors are not correlated, while not imposing this structure. Conversely, the analog covariance structure in the PARAFAC decomposition imposes independence among all factors. 
Overall, our proposed approach has the advantage of promoting a more effective data compression than a PARAFAC, but at the same time, it is a parsimonious representation as many entries of the core tensor shrink toward zero.  

The hierarchical Bayesian model specification is completed specifying a prior distribution for the error variance in \eqref{eq:vecmodel}. For simplicity we let $\sigma^2 \sim \text{InvGa}(a_{\sigma}, b_{\sigma})$. 
When $\Y$ has $m<n$ missing entries, the likelihood can be easily adapted from  \eqref{eq:vecmodel} by removing, from the covariance matrix, the associated $m$ rows and columns. In this case, our Bayesian approach seamlessly allows to perform missing data imputation by simulating from the posterior predictive distribution. 
Extending the cumulative shrinkage process framework to predict missing values is a key contribution that sets apart our work from the original CUSP \citep{CUSP} for matrix data.  See Sections \ref{sec:simulations} for compelling results on tensor completion using synthetic data, as well as an image reconstruction application.

\subsection{Extension for binary tensor data} \label{sec:binProbit}
AdaTuck can be easily adapted to the case of binary tensor data exploiting a probit model. We refer to this specification as probit-AdaTuck.
In Section \ref{sec:fish} we will apply the probit-AdaTuck model for analyzing spatio-temporal species co-occurrence data in a biodiversity study in the North Sea. Other examples of binary multiway data include multirelational social networks \citep{nickel2011three} and connection/disconnection in
brain structural connectivity networks \citep{wang2019common}.

The probit AdaTuck model is defined through a latent  construction. Let $\Q \in \R^{n_1 \times \cdots \times n_K}$ denote a real valued tensor having the same dimensions as $\Y$, with tensor elements $q_{\mathbf{i}}$ for $\mathbf{i} \in I$ and let $y_{\mathbf{i}}$ be $0$ or $1$ according to the sign of $q_{\mathbf{i}}$. We assume
\begin{equation} \label{eq:binaryTuck}
    y_{\mathbf{i}} = \mathds{1}(q_{\mathbf{i}}>0), \quad   \Q = \Z+ \mathcal{E}, \quad \Z = \sum_{r_1=1}^{R_1} \dots \sum_{r_K=1}^{R_K} g_{r_1 \dots r_K} \uu[1]{r_1} \circ  \cdots  \circ \uu[K]{r_K},
\end{equation}
where $\tvec(\mathcal{E}) \sim N_n(0, \sigma^2 I_n)$ and we assume $\sigma^2=1$ as standard practice for probit likelihood. 
For the factor matrices and the core tensor we consider the same prior distributions defined in \eqref{eq:nu} and \eqref{eq:corePar}.
Integrating out the latent variables $q_{\mathbf{i}}$ we have $y_{\mathbf{i}} \sim \mathrm{Bernoulli}(\pi_{\mathbf{i}})$ with
\begin{equation} \label{eq:pimarginal}
    \pr(y_{\mathbf{i}}=1 \mid z_{\mathbf{i}}) =  \pi_{\mathbf{i}} = \Phi(z_{\mathbf{i}}),
\end{equation}
where $\Phi(\cdot)$ is the Gaussian cumulative distribution function and $\pi_{\mathbf{i}}$ is an element of the probability tensor $\mathcal{M}$ with dimension $n_1 \times \cdots \times n_K$.
The latent variable construction of the probit model in \eqref{eq:binaryTuck} is particularly useful for posterior inference, as it enables the use of an efficient data augmentation scheme \citep{albert1993bayesian}.

\subsection{Consistency in rank selection}

One of the key aspects of AdaTuck is its ability to learn the rank along  each  dimension, through the multiway increasing shrinkage prior. Consistent with this characteristic, in this section, we investigate the asymptotic properties of the adopted multiway CUSP in a simplified setting. This analysis extends the discussion in Section 3.3 of  \cite{kowal2023semiparametric} and relates to the problem of estimating a high-dimensional mean vector from a single multivariate observation \citep{castillo2012, sparseRockova}. While these theoretical results are derived under simplified assumptions and do not directly apply to the situations discussed in previous sections, they provide valuable insights for the proposed approach. The asymptotic analysis serves as a principled justification for the multiway CUSP’s ability to discard superfluous dimensions in controlled settings, offering theoretical support for its rank-learning capability. The empirical results presented in Section~\ref{sec:illustration}, which focus on the actual model discussed earlier, demonstrate that the method maintains its consistency in rank estimation, even when applied beyond the simplified framework of this section.

Let $y=(y_1, \dots, y_n)$ with 
\begin{equation} \label{eq:simplelik}
    y_i = \prod_{k=1}^K \uu[k]{i} + \varepsilon_i, \quad \varepsilon_i \sim N(0,1), \quad u_{i}^{(k)} \sim N(0, \theta^{(k)}_{i}),
\end{equation}
with $\theta^{(k)}_{i}$ following the prior structure in \eqref{eq:nu}. While this setting differs from the Tucker decomposition model, we observe that the above likelihood shares some similarities with the proposed models as it corresponds to the multiplicative term in the definition of the elementwise Tucker decomposition in \eqref{eq:Tuck_elmw}.
Hence $y_i \sim N(\eta_i,1)$ for $i=1,2,\dots$ with $\eta_i = \prod_{k=1}^Ku_i^{(k)}$.
We assume that the true data generating process is 
\begin{equation} \label{eq:truelik}
    y_i =  \prod_{k=1}^K u^{0(k)}_{i} + \varepsilon_i,  \quad \varepsilon_i \sim N(0,1), 
\end{equation}
where $u^{0(k)}_{1}, \dots, u^{0(k)}_{n}$ are the true parameters and $R_k^{0n}$ denote the true number of non-zero elements for $k=1,\dots,K$. 
In this setting $u^{0(k)}_{i} \neq 0$ for $i \le \Tilde{R}^{0n}$ and $u^{0(k)}_{i} = 0$ for $i > \Tilde{R}^{0n}$, with $\Tilde{R}^{0n}= \mathrm{min}_{k \in \{1,\dots,K\}} R_k^{0n}$.
We allow the quantities $R_k^{0n}$, that correspond to the true ranks composing the Tucker multi-rank $(R_1^{0n}, \dots ,R_K^{0n})$, to grow with the dimension of the tensor $n$, while $K$ is fixed. 
Taking the minimum among the $K$ ranks ensures the identification of the number of non-zero elements under the data generating process \eqref{eq:simplelik} even though the factors in the product of $u_i^{0(k)}$ remain unidentifiable.

Define the quantity $M_n = \sum_{i \ge \Tilde{R}^{0n}} \prod_{k=1}^K \omega_{i}^{(k)}$, which represents an upper bound for the probability mass assigned to the multiway CUSP slabs for the parameters that are null.
First, we want to show that $M_n$ is far from zero with vanishing probability \emph{a priori}.

\begin{lemma} \label{lm:prior}
    Let $\epsilon_n \to 0$ with $\epsilon_n^{1/\Tilde{R}^{0n}} > \alpha^{(k)}/(1+\alpha^{(k)})$ for any $k=1\dots,K$. For a positive constant $C$, the remainder $M_n$ satisfies $\pr(M_n > \epsilon_n) \le \exp (-CK\Tilde{R}^{0n})$.
\end{lemma}

According to Lemma \ref{lm:prior}, the prior probability that $M_n$ exceeds a small threshold is exponentially small. To connect the prior behavior in Lemma \ref{lm:prior} to the posterior we denote with $\pr_0(\cdot)$ and $E_0(\cdot)$ the probability and the expectation respectively under the true distribution of the data, described in  \eqref{eq:truelik}. Theorem \ref{thm:consistency} shows that the posterior probability that $M_n$  exceeds a shrinking $\epsilon_n$ for $n \to \infty$ goes to zero in expectation thus providing theoretical justification for the use of the proposed multiway CUSP to learn the ranks of model \eqref{eq:tuck_intro}. Intuitively, this implies that the posterior distribution asymptotically concentrates on models where redundant parameters---those not contributing to the true data-generating structure---are discarded, ensuring consistent recovery of the tensor ranks as the sample size increases. 

\begin{theorem} \label{thm:consistency}
     Let $\epsilon_n \to 0$ with $\epsilon_n^{1/\Tilde{R}^{0n}} > \alpha^{(k)}/(1+\alpha^{(k)})$ for any $k=1\dots,K$. We assume $A\Tilde{R}^{0n}/n<1/2$ with $A>1/2$. Let $\theta_\infty = \theta_{\infty}(n)$ in the multiway CUSP prior \eqref{eq:nu} with $\theta_\infty(n) <  ({\tilde R}^{n0}/n)^{1/(2K)}$. Then, $\lim_{n \to \infty} E_0\{\pr(M_n > \epsilon_n \mid y_1,\dots, y_n)\} = 0.$
\end{theorem}

\section{Posterior computation}
\label{sec:computation}

Posterior inference for AdaTuck  proceeds via Gibbs sampling. 
Specifically, the sampler exploits a convenient data augmentation of the multiway CUSP prior. Let $\sk[k]{r_k} \in \{1, \dots, \infty\}$ denote independent categorical variables with $P(\sk[k]{r_k}  = l \mid \omega_{l}^{(k)}) = \omega_{l}^{(k)}$ ($l = 1, \dots, R_k$). The specification
\begin{equation*} 
     \theta_{r_k}^{(k)} \mid \sk[k]{r_k} \sim \{1 - \mathds{1}(\sk[k]{r_k} \le r_k) \} \text{InvGa}(a_{\theta}, b_{\theta}) + \mathds{1}(\sk[k]{r_k} \le r_k ) \delta_{\theta_{\infty}}
\end{equation*}
induces (\ref{eq:nu}) via marginalization over $\sk[k]{r_k}$. The number of active (slab) terms are therefore $R^{*}_k = \sum_{r_k=1}^{\infty} \mathds{1}(\sk[k]{r_k} > r_k )$. In other words, $R^{*}_k$ are the effective number of nonparametric terms composing the multi-rank. We start by introducing the Gibbs sampler for a fixed truncation level $(R_1, \dots, R_K)$, as formally justified in Theorem \ref{thm:truncation}, and then propose an adaptive strategy for inference on $R_k^*$ with $k=1,\dots, K$. 

Conditioned on $\sk[k]{r_k}$ ($k=1, \dots, K$ and $r_k = 1, \dots, R_k$) it is possible to sample from conjugate full-conditionals.
The joint distribution of $\theta_{r_k}^{(k)}$ and $\uu[k]{r_k}$ is a $r_k$ dimensional multivariate normal centered in zero with covariance matrix $\theta_{r_k}^{(k)}I_{n_k}$.
Marginalizing out $\theta_{r_k}^{(k)}$ from the latter distribution, one can obtain the 
full-conditional distributions for the updating of the augmented data. They are
\begin{equation} \label{eq:Sh_multinomial}
    P(\sk[k]{r_k} = l) \propto 
    \begin{cases}
    \omega_{l}^{(k)} N_{n_k}(\uu[k]{r_k}; 0, \theta_{\infty}I_{n_k}), &\text{for $l=1, \dots, r_k$},\\
    \omega_{l}^{(k)} t_{2a}(\uu[k]{r_k}; 0, (b_{\theta}/a_{\theta}) I_{n_k}), &\text{for $l=r_k+1, \dots, R_k$},\\
    \end{cases}
\end{equation}
where  $N_p(x;\mu, \Sigma)$ and $t_v(x;\mu,\Sigma)$ are respectively the densities of $p$-variate Gaussian and Student-$t$ with $v$ degrees of freedom distributions evaluated at $x$.
For the updating of the factor matrices we exploit the matricized form of the Tucker decomposition. Let $Z_{(k)}$ denote the mode-$k$ matricization of a tensor $\mathcal{Z}$ that arranges the mode-$k$ fibers to be the columns of the resulting matrix. The Tucker decomposition of $\Z$ can be equivalently written as
\begin{equation*}
    Z_{(k)} = \UU{k} G_{(k)} W^{(-k)}, \quad \text{with } W^{(-k)} = (\UU{K} \otimes \cdots \otimes \UU{k+1} \otimes \UU{k-1} \otimes \cdots \otimes \UU{1})^{\top},
\end{equation*}
where $G_{(k)}$ is the mode-$k$ matricization of the core tensor $\G$.
Algorithm \ref{algo:gibbs} summarizes all the steps of the Gibbs sampler with a finite truncation level for the multi-rank.

\begin{algorithm}[t]
    \caption{One cycle of the Gibbs sampler for the AdaTuck model with finite truncation level $(R_1, \dots, R_K)$.} \label{algo:gibbs}
    \begin{enumerate}

    \item  For $k$ from 1 to $K$ and for $i_k$ from 1 to $n_k$:\\
     $\quad$ sample the $i$-th row of $\UU{k}$ from   $N_{R_k}(V_u G_{(k)} W^{(-k)} \sigma^{-2} y_{i_k}, \ V_u)$, with $y_{i_k}$ the $i_k$-th\\ $\quad$ row of $Y_{(k)}$ and
    $V_u = \{\tdiag(\theta_{1}^{(k)}, \dots, \theta_{R_k}^{(k)})^{-1} + G_{(k)} W^{(-k)} \sigma^{-2}  W^{(-k) \top} G_{(k)}^{\top}\}^{-1}$. 

    \item Sample $\tvec(\G)$ from $N_R(V_g W^{\top} \sigma^{-2} \tvec(\Y), V_g )$, with $V_g = (\Psi^{-1} +\sigma^{-2} W^{\top}  W)^{-1}$, 
     $\Psi = \text{diag}(\psi_{1\cdots1}, \dots, \psi_{R_1 \cdots R_K})$ and $\psi_{r_1 \cdots r_K} = \tau\nu_{r_1 \cdots r_K}$.
    
    \item  Sample $\sigma^2$ from 
    InvGa$(a_{\sigma} + 0.5 n, \ b_{\sigma} +0.5 \sum_{k=1}^{K} \sum_{i_k=1}^{n_k}(y_{\mathbf{i}} - z_{\mathbf{i}} )^2$,  with $z_{\mathbf{i}} = \sum_{k=1}^{K} \sum_{r_k=1}^{R_k}  g_{\mathbf{r}} \uu[1]{r_1} \circ \cdots  \circ \uu[K]{r_K}$.

    \item Sample $\tau \sim \mathrm{GIG}(a_{\tau} - 0.5R, \sum_{k=1}^K\sum_{r_k=1}^{R_k} g_{\mathbf{r}}^2/ \nu_{\mathbf{r}}, 2b_{\tau})$, where GIG is shorthand for the Generalized Inverse Gaussian distribution.

    \item For $\mathbf{r} \in J$, with $J$ the index set of all observations of $\G$:\\
    $\quad$ sample $\nu_{\mathbf{r}} \sim \mathrm{GIG}(0.5, g_{\mathbf{r}}^2/\tau, \rho^2_{\mathbf{r}})$.

    \item For $\mathbf{r} \in J$:\\
    $\quad$ sample $\rho_{\mathbf{r}} \sim Ga(a_{\rho}+1, b_{\rho} + |g_{\mathbf{r}}|/\tau^{1/2})$.
    
    \item For $\mathbf{r} \in J$:\\ 
        $\quad$ sample $\sk[k]{r_k}$ from the categorical distribution with probabilities as in (\ref{eq:Sh_multinomial}).

    \item For $k$ from 1 to $K$ and for $l$ from 1 to $R_K$:\\ 
    $\quad$ update $V_{l}^{(k)}$ from
    Be$(1 + \sum_{r_k=1}^{R_k} 
    \mathds{1}(\sk[k]{r_k} = l), \alpha^{(k)} + \sum_{r_k=1}^{R_k}\mathds{1}(\sk[k]{r_k} > l))$.\\
    Set $V_{k,R_k}=1$, and update $\omega_{k,1}, \dots, \omega_{k,R_k}$ through (\ref{eq:nu}).

  \item For $k$ from 1 to $K$ and for $r_k$ from 1 to $R_k$:\\
    $\quad$ if $\sk[k]{r_k} \le r_k$ set $\theta_{r_k}^{(k)} = \theta_{\infty}$; otherwise sample $\theta_{r_k}^{(k)}$ from\\ $\quad$ InvGa$(a_{\theta}+0.5n_k, \ b_{\theta} + \sum_{i=1}^{n_k} u^{(k) \ 2}_{ir_k})$.    
    \end{enumerate}

\end{algorithm}

It is reasonable to run the Gibbs sampler in Algorithm \ref{algo:gibbs} using a sufficiently large truncation level $(R_1,\dots,R_K)$, with $R_k \le n_k +1$ since there are at most $R_k  -1$ active components for each mode. However, in practice we expect relatively few important factors compared with the dimension of the original tensor.
In this regard, it is crucial to design an appropriate computational strategy that balances the risk of missing important factors by choosing a truncation level that is too small, and the inefficiency that comes with setting the truncation level too conservatively.

We tackle this issue by implementing an adaptive  Gibbs sampler that tunes the latent multi-rank as the sampler proceeds, building on similar strategies employed in \cite{CUSP} and  \cite{Bhattacharyadunson11}. 
Specifically, we adapt $(R_1, \dots, R_K)$ at  iteration $t$ with probability $\exp(c_0 + c_1t)$ with $c_0 \le 0$ and $c_1 < 0$, to satisfy the diminishing adaptation condition in \cite{roberts2007coupling}. The adaptation consists in dropping the inactive columns of $U^{(k)}$ for $k=1,\dots,K$, if any, together with the corresponding parameters and the corresponding rows of the mode-$k$ matricization of the core tensor $G_{(k)}$.
Consequently, we decrease the truncation level by setting $R_k = R_k^*+1$, with  $R^{*}_k = \sum_{r_k=1}^{R_k} \mathds{1}\{\sk[k]{r_k} > r_k \}$.
If all columns of $U^{(k)}$ are active an extra factor for dimension $k$ is added, sampling the associated parameters from the prior. The inactive columns of $U^{(k)}$ are naturally identified under \eqref{eq:nu} since they are those modelled by the spike and hence have indexes $r_k$ such that $S_{r_k}^{(k)} \le r_k$.
This strategy allows to independently adapt the elements for each mode, such that at iteration $t$  $R_k$ may decrease while $R_{k'}$ may not. Algorithm \ref{algo:adaptivegibbs} shows the detailed steps of the adaptive Gibbs sampler.

\begin{algorithm}[t]
    \caption{One cycle of the adaptive Gibbs sampler for the AdaTuck model.} \label{algo:adaptivegibbs}

Let $R^{(t)}_k$ be the truncation at iteration $t$ for mode $k$ and  $R^{*(t)}_k = \sum_{r_k=1}^{R_k^{(t)}} \mathds{1}\{
S^{(k)(t)}_{r_k} > r_k \}$.
\begin{enumerate}
    \item Perform one cycle of Algorithm \ref{algo:gibbs}.
    \item if $t \ge \Tilde{t}$ adapt with probability $p(t) = \exp(c_0 + c_1t)$ as follows\\
    for $k$ from $1$ to $K$:\\
    \begin{itemize}
        \item \textbf{if} $R^{*(t)}_k < R^{(t-1)}$\\
        set $R^{(t)}_k = R^{*(t)}_k+1$, drop the inactive columns in $U^{(k)}$ together with the associated parameters in  $\theta^{(k)}$, $\omega^{(k)}$ and the mode-$k$ matricization of $\nu$, $\rho$ and $\G$ and add a final components to  $U^{(k)}$, $\theta^{(k)}$, $\omega^{(k)}$ and the mode-$k$ matricization of $\nu$, $\rho$ and $\G$ sampled from the corresponding priors;
        \item \textbf{else}\\
        set $R^{(t)}_k = R^{(t-1)}_k +1$ and add a final column sampled from the spike to $U^{(k)}$, together with the associated parameters in  $\theta^{(k)}$, $\omega^{(k)}$ and the mode-$k$ matricization of $\nu$, $\rho$ and $\G$, sampled from the corresponding priors. 
    \end{itemize}

    \end{enumerate}
\end{algorithm}

The proposed approach provides a principled manner for dealing with multiway data with missing records. In the Gibbs sampler, indeed, it is straightforward to sample conditioning only on the observed entries of the data. Then, for predicting purposes, we can impute the missing entries of $\Y$ leveraging the posterior samples of the latent quantities of the factorization. 
Posterior inference in the probit-AdaTuck proceeds via an analog adaptive Gibbs sampler that leverages a probit data augmentation scheme \citep{albert1993bayesian}. 
See Section 2 of the Supplementary Materials for practical implementation specifics and the detailed steps of the adaptive Gibbs sampler for binary tensor data.

\section{Empirical illustrations} \label{sec:illustration}

\subsection{Synthetic data}
\label{sec:simulations}

In this section, we evaluate the performance of AdaTuck with synthetic data. Specifically, we assess its performance in terms of learning the true latent multi-rank of the Tucker factorization, under correct specification. Additionally, we assess its goodness-of-fit by randomly removing some entries of the tensor data before fitting and then imputing these missing values by exploiting the posterior predictive distribution. See Section 3 of the Supplementary Materials for additional experiments in data completion with real image data.

In line with the real data applications presented in the following sections, we focus on three-way tensor, i.e. $K = 3$, and we simulate $\Y$ assuming the Tucker factorization structure of equation (\ref{eq:tuck_intro}).
The rows of the factor matrices $\UU{K}$ are drawn from independent normal centered in zero with a diagonal covariance matrix whose elements are drawn from independent inverse gamma distributions with both parameters equal to two, while the entries of the residual noise tensor $\Z$ are drawn from independent $N(0,0.1)$. We simulate the entries of $\G$ from independent $N(0,1)$ and introduce local sparsity by setting equal to zero $40\%$ of the total core tensor entries.
We consider two different settings for the tensor dimension, one where the latent ranks are different and so the Tucker factorization is effective, i.e. $(n_1, n_2, n_3) =(50, 40, 6)$ and $(R_1^{0}, R_2^{0}, R_3^{0}) =(10, 7, 3)$, and one where the data are close to a PARAFAC decomposition model, i.e. $(n_1, n_2, n_3) =(30, 30, 10)$ and $(R_1^0, R_2^0, R_3^0) =(5, 5, 5)$. To assess the performance of our proposed approach under different amounts of data availability, we consider two different settings for missing data. For both scenarios we randomly impute to missing values the $30\%$ and the $50\%$ of the total entries of the data $\Y$. In each setting we simulate $B=20$ different synthetic datasets.

We perform posterior inference on $\Y$ under our proposed approach exploiting the adaptive Gibbs sampler in Algorithm \ref{algo:adaptivegibbs}. We set $(a_{\theta}, b_{\theta}) = (a_{\tau}, b_{\tau}) =(2, 2)$, $\alpha^{(k)} = 3$ for each $k=1,\dots,K$,  $(a_{\sigma}, b_{\sigma}) = (1, 0.3)$,  $(a_{\rho}, b_{\rho})=(10,10)$ and $\theta_{\infty} = 0.05$ following \cite{CUSP}. The Gibbs sampler is run for $12{,}000$ iterations of which the first $8{,}000$ are discarded as burn-in.

To assess the performance of our proposed approach in learning the true underlying multi-rank $(R_1^0, R_2^0, R_3^0)$, we compute the
Monte Carlo average over the $B$ replicates of the posterior median of $R_1^*$,  $R_2^*$ and  $R_3^*$ under the AdaTuck model. 
The results for the $B$ different datasets across the two different scenarios are shown in Figure \ref{fig:rank}. In both scenarios, our proposed method correctly recovers the true multi-rank for all tensor margins, thus efficiently learning the truncation level in the adaptive Gibbs sampler.  AdaTuck posterior estimations of the multi-rank highly concentrate around the true values $(R_1^0, R_2^0,R_3^0)$, and the results are robust with respect to the number of missing values in the data.
Figure S2 in the Supplementary Materials illustrates the results for the probit-AdaTuck model in the same scenarios, showing its effectiveness in recovering the true multi-rank  also in the binary data setting.
For continuous tensor data, we benchmark multi-rank estimations using the Eigenvalue Ratio (ER) and the Growth Ratio (GR) criteria  \citep{ahn2013eigenvalue}, applied to each mode-$k$ matricization 
$Y(k)$. These methods perform well on fully observed tensors but fail for partially observed data, as they require complete matrices for eigendecomposition. When missing values are naively imputed using the mean, ER and GR severely underestimate the true multi-rank, as shown in Table S2 of the Supplementary Materials.

\begin{figure}[t]
\begin{center}
    \includegraphics[width=0.9\textwidth]{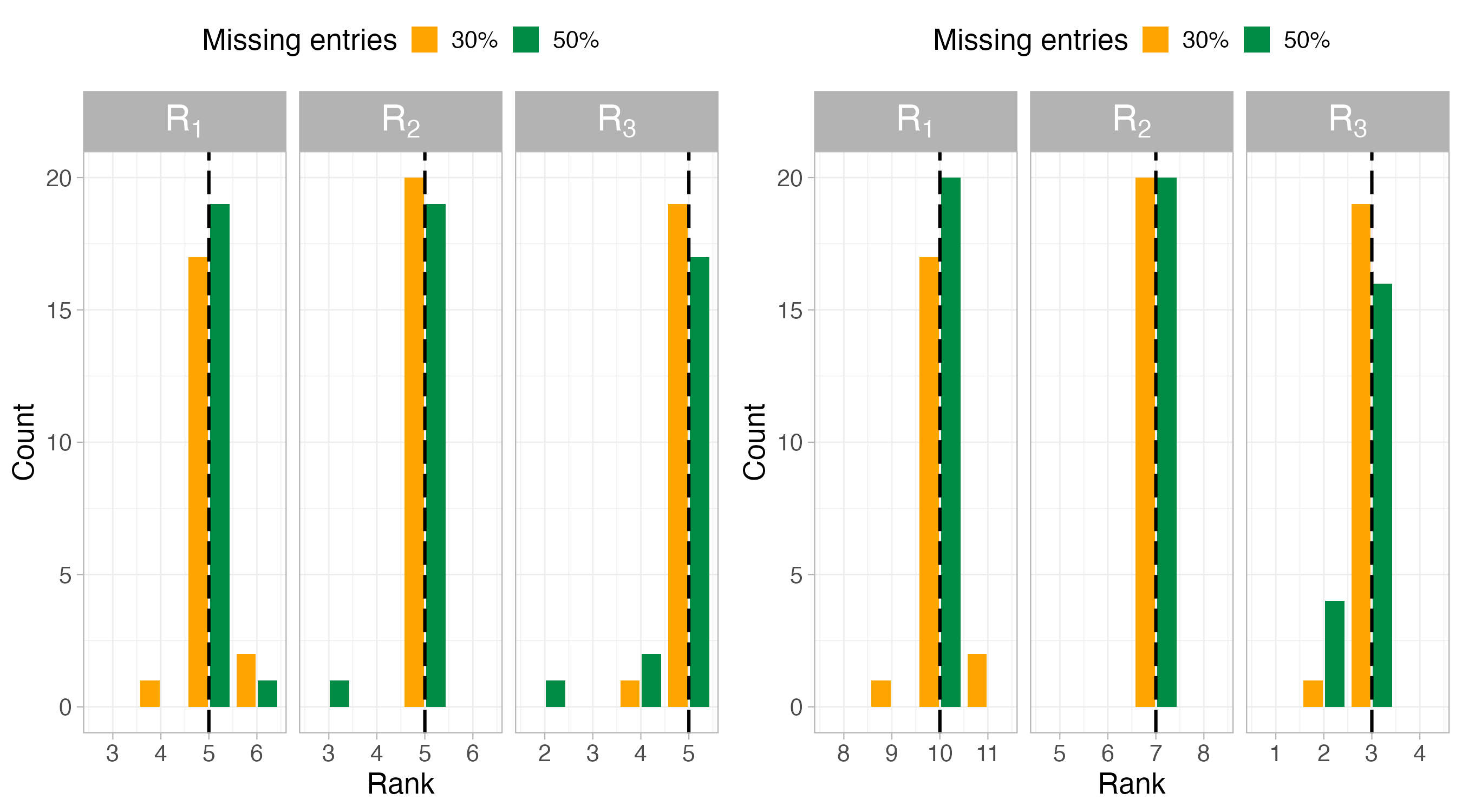} 
    \caption{Distribution of the Monte Carlo estimations for the median of $R_1^*$, $R_2^*$ and $R_3^*$ for the 20 different simulations in the two scenarios for $30\%$ and $50\%$ of total entries missing. Black dotted lines are the true $R_1^0$, $R_2^0$ and $R_3^0$ values.}
    \label{fig:rank}
\end{center}
\end{figure}

Subsequently, we evaluate the performance of the proposed  method in predicting the missing values and reconstructing the original tensor for the two scenarios. We benchmark the results with the  Bayesian tensor factorization method proposed in \cite{khan2016tensorbf} and implemented in the R package \texttt{tensorBF}, which targets multiway data with incomplete records. This method implements a Bayesian PARAFAC decomposition and achieves rank optimization through a sparsity parameter to remove redundant factors. We run \texttt{tensorBF} with the default options for the same number of iterations we used for AdaTuck. 
We compare the two methods by computing the  Monte Carlo mean square error (MSE) in predicting the missing values with the reconstructed tensor based on the latent structure. 
Table \ref{tab:simMSE} summarizes, for each scenario and model, the median and the interquartile range of the MSE obtained from the $B$ measures generated from different simulations, together with the running times. Such quantities rely on an \texttt{R} implementation run on a Apple M1 CPU laptop with 8
GB of RAM.
Across all scenarios AdaTuck outperforms the competitor in predicting the missing entries of the tensor. As expected, \texttt{tensorBF}, which implements a PARAFAC decomposition, performs better in the scenario similar to the PARAFAC structure. Instead,  AdaTuck with its greater flexibility provides good results in both scenarios. Additionally, AdaTuck is robust to variations in the amount of available data, maintaining stable performance even as the percentage of missing values increases, whereas the competitor's performance gets worse as missing entries increase.

\begin{table}[t]
    \centering
    \caption{Performance of AdaTuck and TensorBf in predicting the missing entries for the 20 simulated datasets for the two scenarios. The shorthand `NA(\%)' represents the percentage of missing entries with respect to the total tensor entries. The runtime is in seconds.}
    \vspace{0.3cm}
    \begin{tabular}{cc|ccc|ccc}
    
   \multicolumn{2}{}{} &
      \multicolumn{3}{c}{AdaTuck} &
      \multicolumn{3}{c}{ TensorBF}  \\
      
    $(R_1^0, R_2^0, R_3^0)$ & NA (\%) & MSE  &  MSE  & runtime & MSE  &  MSE   &runtime \\
    && (median) & (IQR) & (median) & (median) & (IQR) & (median)\\
    \hline
   (5,5,5) & 30  & 0.22 & 0.01& 1053 &33.00 & 47.27 & 240\\
    & 50 & 0.23  & 0.02 & 962 & 46.73& 374.78 & 241\\
    (10,7,3) & 30  & 0.23 & 0.01 & 2234 & 138.27 & 121.13&  326\\
    & 50  & 0.24  & 0.01 & 1781 & 146.28& 153.58 & 324\\
    \hline 
  \end{tabular}
    \label{tab:simMSE}
\end{table}

\subsection{Licorice chemometrics data}
Chemical data are often multiway by nature, making tensor factorization techniques popular in chemometrics \citep{tomasi2005parafac}.
In this section, we analyze a sensor-based dataset from the licorice industry \citep{skov2005new}, driven by the food industry’s increasing focus on process control and quality checking.
Data from sensor-based analytical instruments (e.g. electronic nose) can be arranged in a three-way array as sample $\times$ time $\times$ sensor.
The electronic nose is used to distinguish complex volatiles, which can reproduce the structure and principle of olfactory sense.
The dataset consists of $6$ good licorice samples and $12$ bad licorice samples, half of which are artificially fabricated by drying for a long time good samples.
Twelve sensors, all based on Metal Oxide Semiconductor technologies,  were used to capture the volatile compound measurements every 0.5 seconds over a span of 120 seconds, resulting in a data tensor $\Y$ of dimensions $18 \times 241 \times 12$.
The sensors are distributed in two chambers with different temperatures.
The data are normalized over all cells to have mean zero and unit variance, allowing one to assume default values for hyperparameters in the proposed multiway priors.

We apply AdaTuck to the Licorice data, using the same hyperparameters and number of iterations for the Gibbs sampler described in Section \ref{sec:simulations}. 
The Markov chain is then thinned by retaining every $4$-th sample, resulting in a total of $2{,}000$ samples. 
The Monte Carlo mean of the posterior median of the latent multi-rank is $(2, 34, 7)$. Notably, $R_1^*$, the latent rank corresponding to the samples mode, was estimated as two, capturing the distinction between good and bad quality licorice. 

To investigate whether bad and good licorice samples exhibit distinct behaviors, we examine posterior samples from AdaTuck.
Figure \ref{fig:postChem} shows signal posterior means and credible intervals for four different licorice samples (two bad and two good) for three different sensors.
The plot shows the typical sensor response curves which first increase (adsorption phase) reaching the maximum peak and then decrease returning to the initial stage (desorption phase).
Across all sensors, bad and good samples follow similar temporal trends but consistently exhibit different signal levels, clearly delineating the two quality groups. This difference is particularly evident at the maximum peak of the curves, a critical part for sample separation as also discussed by \cite{skov2005new}.

\begin{figure}[t]
\begin{center}
    \includegraphics[width=0.99\textwidth]{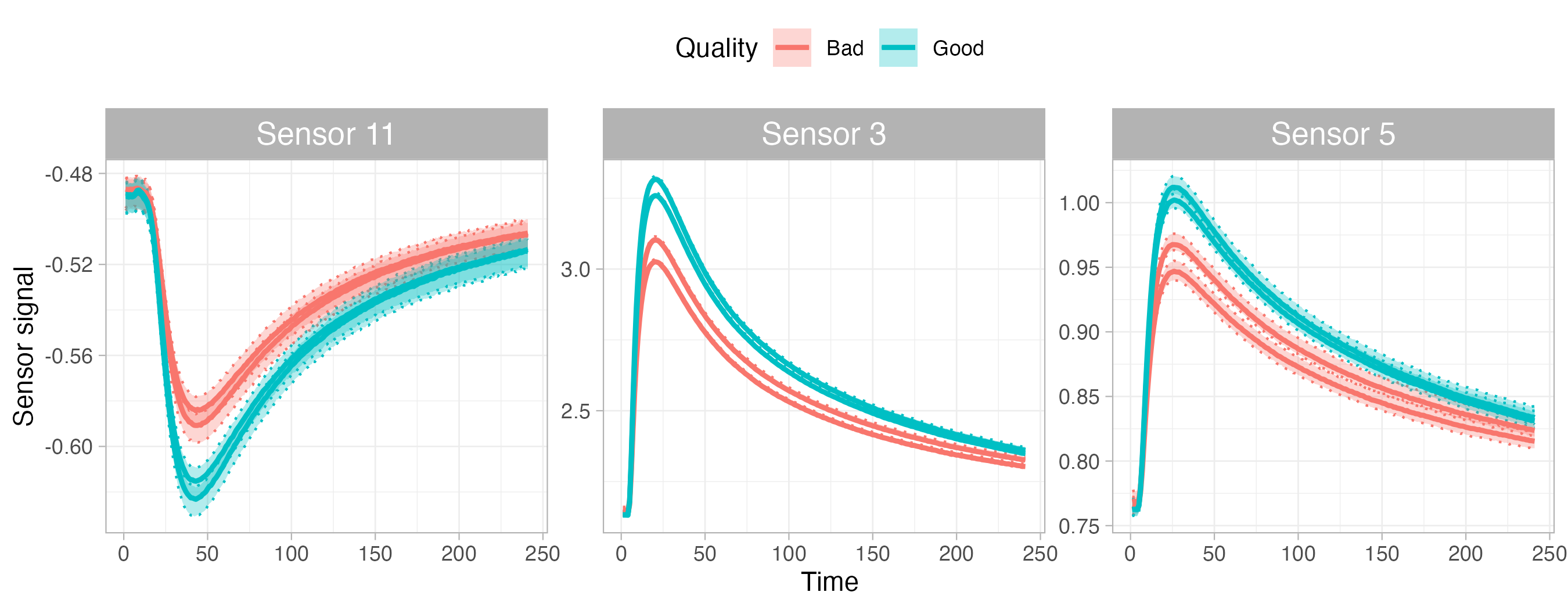} 
    \caption{Signal posterior means and $90\%$ credible intervals for four different licorice samples (two good and two bad) for three different sensors.}
    \label{fig:postChem}
\end{center}
\end{figure}

For the three sensors in Figure \ref{fig:postChem} the time profiles reveal differences in magnitude, 
while their overall shapes remain similar. 
This behavior is consistent across all 18 sensors and can be further analyzed by examining the posterior samples
of $U^{(3)}$. Figure S3 in the Supplementary Materials  shows a representative sample of $U^{(3)}$, the factor matrix associated to the sensor margin, randomly selected from the Markov chain Monte Carlo iterations, along with the posterior means of the sensor response curves for the twelve sensors for one licorice sample. Different samples from the Markov chain lead to qualitatively similar results. 
The posterior means of sensor response curves show two distinct groups: one with negative magnitudes with respect to the global mean and one with positive magnitudes. Within the positive group, sensor 3 stands out with a higher magnitude compared to the other sensors. This pattern is also evident in the posterior sample of $U^{(3)}$ where two groups of positive and negative loadings emerge,  with sensor 3 exhibiting larger values for corresponding entries. These two groups may reflect the placement of the sensors in two separate chambers.

\subsection{Fish biodiversity data} \label{sec:fish}

We analyze data from a community ecology monitoring program of the North Sea demersal fish community. The  International Council for the Exploration of the Sea (ICES) created a unique data set covering more than 30 years of multiple fish species abundance across different locations in the North Sea.  Data have been collected using standard otter trawls as sampling gear.
The North Sea marine ecosystem has already been significantly impacted by climate change and fisheries exploitation, motivating a comprehensive analysis of fish community dynamics to support efforts in preserving this fragile ecosystem.

Fish co-occurrence data with repeated sampling at multiple locations can be naturally organized as a binary three-way tensor in the form of species $\times$ time $\times$ space. The data contains fish co-occurrence for a total of $65$ individual species for the period 1985 to 2015. The data used refer to the first quarter of each year, to avoid seasonality bias.
The North Sea is divided into seven areas, based on ecological aspects of the fish fauna, such as feeding and species composition. Hence, the binary tensor $\Y$ has dimension $65 \times 31 \times 7$. 
These data have been previously analyzed by \cite{Frelat3D}, one of the few contributions in community ecology to preserve the data's intrinsic multi-view structure without reducing it to a matrix and sacrificing one of the dimensions. However, \cite{Frelat3D} do not assume any probabilistic framework or make any distributional assumption, but analyze the data through principal tensor analysis over $k-$modes.

We analyze North Sea fish co-occurrence using the probit-AdaTuck model presented in Section \ref{sec:binProbit}. 
Posterior inference is performed using the same hyperparameters as in Section \ref{sec:simulations}. The adaptive Gibbs sampler is run for $20{,}000$ iterations, with the first $10{,}000$ discarded as burn-in. The Markov chain is then thinned by retaining every $5$-th sample, resulting in a total of $2{,}000$ samples. 
The Monte Carlo mean of the posterior median of the latent multi-rank under the AdaTuck model is $(R_1^*, R_2^*, R_3^*) = (9,3,4)$. 
This result highlights the effectiveness of using the Tucker model for achieving a more efficient data representation compared to the PARAFAC model. 
Given the different margin dimensions of the data, it was reasonable to expect the rank for the first dimension to be higher than the rank for the third dimension.

To investigate the spatial and temporal dynamics of fish communities we look at  the posterior marginal probabilities $\pi_{\mathbf{i}}$ in \eqref{eq:pimarginal} composing the corresponding probability tensor $\mathcal{M}$. 
Analyzing these probabilities across different tensor slices can provide valuable insights into population behavior, potentially revealing interesting spatial or temporal patterns. Figure S4 in the Supplementary Materials shows the posterior mean of the frontal slice of $\mathcal{M}$, $M_{::i}$ for $i=1,\dots,n_3$, for two roundfish areas of the North Sea: the northernmost area and one located in the southeast. The latter highlights the differences between fish species in the southern part of the North Sea and those in the northern region, which is deeper and experiences less seasonal temperature variation. As a result, we observe a strong  north-south component in the composition of the North Sea fish community, consistent with previous findings \citep{Frelat3D}.
Looking at the biogeographic information of species, we observe that Atlantic fishes, species widespread in the North Atlantic, have a small probability of occurring in the southern area close to the Netherlands ' coasts, while they have a higher probability of occurring in the northeast area, as it borders the Atlantic Ocean.
Conversely, Lusitanian fishes have a higher probability of occurring in the southern area. Lusitanian fishes, indeed,  tend to be abundant from the Iberian peninsula to as far north as the British Isles and the central North Sea \citep{engelhard2011ecotypes}.

While some species, such as \textit{Arnoglossus laterna} and \textit{Chelidonichthys cuculus}, show a positive temporal trend in their occurrence probabilities, our main interest  is identifying negative trends that may indicate endangered or vulnerable species. Figure \ref{fig:trend} shows posterior means and posterior $90\%$ credible intervals of two horizontal slices of $\mathcal{M}$ relatives to the species \textit{Squalus acanthias} and \textit{Anarhichas lupus}. Both species exhibit a negative trend in occurrence probabilities over the years. This trend was particularly evident for the southern-eastern North Sea sites (labeled as `S', `SE', `E') for the \textit{Squalus acanthias} and for the northern-western sites (labeled as `C', `NW', `SW') for the \textit{Anarhichas lupus}. The \textit{Squalus acanthias}, commonly known as Spurdogs, are classified in the IUCN Red List of threatened species as vulnerable globally and critically endangered in the Northeast Atlantic. The observed positive growth after 2005 may indicate that conservation policies have been effective in protecting the Spurdog population.
Furthermore, \cite{bluemel2022decline} observe a decline in \textit{Anarhichas lupus} in the North Sea, commonly named wolffish, mainly due to fishing pressure and climate change. 

\begin{figure}[t]
\begin{center}
    \includegraphics[width=0.99\textwidth]{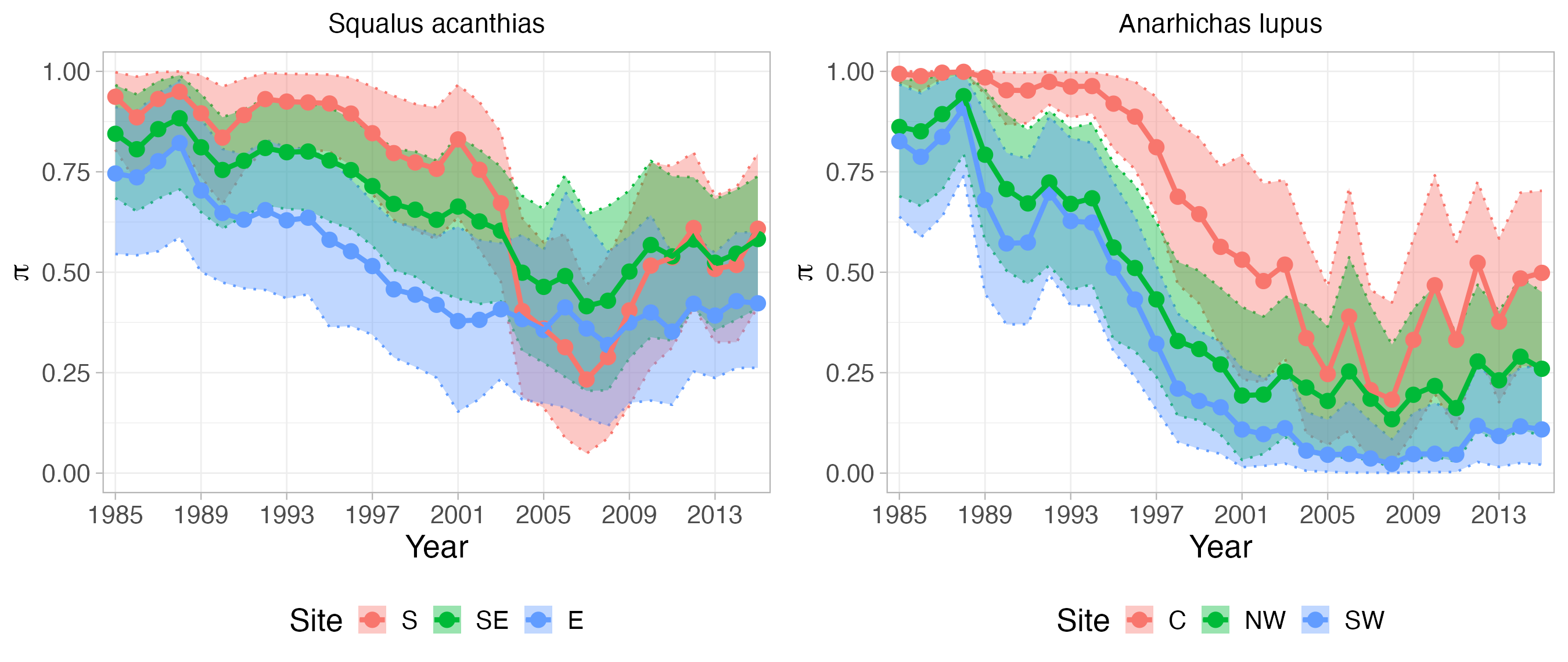} 
    \caption{Posterior mean and $90\%$ credible intervals of two horizontal slices of $\mathcal{M}$: one relative to \textit{Squalus acanthias} for the southern-eastern area of the North Sea (S, SE and E)  and one relative to \textit{Anarhichas lupus} for northern-western areas (C, NW and SW).}
    \label{fig:trend}
\end{center}
\end{figure}

To assess fish interactions, we compute an estimation of the posterior correlation of the latent variable in the probit model. Species that show positive correlations tend to occur together. Recalling from Section \ref{sec:binProbit} that the covariance matrix among the vector of the latent variable $\tvec(\Z)$ is $\Sigma = W \Psi W$, the correlation matrix among $\tvec(\Z)$ is $D^{-1} \Sigma D^{-1}$, where $D$ is the diagonal matrix containing the square root of the diagonal elements of $\Sigma$. Figure S5 in the Supplementary Materials shows the posterior mean of the correlation matrix for the  observations of the first year in the northernmost area. It highlights different positive and negative correlation patterns in fish co-occurrence, potentially identifying connected sub-communities of species.

\section{Discussion} \label{sec:end}
We developed a flexible Bayesian nonparametric framework for analyzing tensor data under a Tucker decomposition model. A key element of our approach is the ability of the model of adapting its complexity: the model  adaptively learns the latent multi-rank based on data characteristics using a multiway shrinkage process. This prior enables principled rank selection through an elegant and coherent probabilistic setting.
Our framework naturally accommodates both continuous and binary data with missing values, using suitable likelihood models. 
We provided theoretical justifications for the proposed AdaTuck model,  with applications demonstrating strong empirical performances.
Code to implement our method and reproduce our analysis is available at \url{https://github.com/federicastolf/AdaTuck}.

A promising direction for the method introduced  would be to extend the proposed unsupervised approaches to the case of linear regression models, where tensor predictors are related to scalar outcomes. Most existing works on this topic are limited to PARAFAC decomposition and do not provide a theoretically justified approach for the rank choice \citep{Guhaniyogietal2017}. Although inference for the proposed method is carried through efficient closed-form Gibbs sampling updates, computational demands can rise with high-dimensional data.
Interesting future directions include developing scalable inference strategies, leveraging approximate Bayesian methods such as expectation propagation or variational Bayes, to further enhance AdaTuck’s applicability to large-scale, multiway data settings.

\subsection*{Acknowledgements}
The authors thank David Dunson, Alberto Cassese and Giovanni Toto for comments on early versions of the draft. 

\bibliography{reference}

\clearpage

\section*{Supplementary Materials}

\setcounter{section}{0}
\setcounter{table}{0}
\setcounter{figure}{0}
\setcounter{equation}{0}
\renewcommand{\thefigure}{S\arabic{figure}}
\renewcommand{\thesection}{S.\arabic{section}}
\renewcommand{\thetable}{S\arabic{table}}
\renewcommand{\theequation}{S.\arabic{equation}}

\section{Proofs}

\begin{proof}[Proof of Theorem 1]
 Let $M_{\textbf{i}}^{R_1,\dots,R_K} = \sum_{\mathbf{r}, r_k = R_k+1}^{\infty} g_{\mathbf{r}} \prod_{k=1}^K \uu[k]{i_k r_k} $ be the residuals for a generic index vector $\mathbf{i}$. By Markov's inequality 
 \begin{equation*}
      \pr\{(M_{\textbf{i}}^{R_1,\dots,R_K})^2 \ge \epsilon\} \le \frac{E\{(M_{\textbf{i}}^{R_1,\dots,R_K})^2\}}{\epsilon},
 \end{equation*}
 and since $g_{\mathbf{r}}$ and $\{\uu[k]{i_k r_k}\}_{k=1}^K$ are independent for all indexes $\mathbf{r}$, the second moment  of $M_{\textbf{i}}^{R_1,\dots,R_K}$ is equal to
\begin{equation*}
     E\{(M_{\textbf{i}}^{R_1,\dots,R_K})^2\} = \sum_{\mathbf{r}, r_k = R_k+1}^{\infty} E(g_{\mathbf{r}}^2) \prod_{k=1}^K E\{(\uu[k]{i_k r_k})^2\},
\end{equation*}
where $E\{(\uu[k]{i_k r_k})^2\} = E[E\{(\uu[k]{i_k r_k})^2 \mid \theta^{(k)}\}] $$= {E(\theta^{(k)}_{r_k})} = \{\alpha^{(k)}(1+\alpha^{(k)})^{-1}\}^{r_k}$ as shown in \cite{CUSP}. Under the assumption that the 
elements of the core tensor $g_{\mathbf{r}}$ have bounded variances, 
$E(g_{\mathbf{r}}^2)$ will be bounded by $\zeta < \infty$. Therefore, the following bound holds
\begin{equation*}
    E\{(M_{\textbf{i}}^{R_1,\dots,R_K})^2\} < \zeta \sum_{\mathbf{r}, r_k = R_k+1}^{\infty} \prod_{k=1}^K \{\alpha^{(k)}(1+\alpha^{(k)})^{-1}\}^{r_k}.
\end{equation*}
The above equation can be factorized with respect to the indexes of the sum, and  by noticing that $\sum_{r_k = R_k+1}^{\infty} \{\alpha^{(k)}(1+\alpha^{(k)})^{-1}\}^{r_k} = \alpha^{(k)} \{\alpha^{(k)}(1+\alpha^{(k)})^{-1}\}^{R_k}$ for any $k=1,\dots,K$ we obtain   
\begin{equation*}
    E\{(M_{\textbf{i}}^{R_1,\dots,R_K})^2\} <  \zeta  \prod_{k=1}^K \alpha^{(k)}\{\alpha^{(k)}(1+\alpha^{(k)})^{-1}\}^{R_k},
\end{equation*} 
which concludes the proof.

\end{proof}

\begin{proof}[Proof of Lemma 1]
By rewriting the remainder as $M_n = \prod_{l=1}^{\Tilde{R}^{0n}} \prod_{k=1}^K (1-V_{l}^{(k)})$, its a priori expectation is
\begin{equation*}
    E(M_n) = \prod_{l=1}^{\Tilde{R}^{0n}}  \prod_{k=1}^K E(1-V_{l}^{(k)}) = \prod_{k=1}^K \left( \frac{\alpha^{(k)}}{1+\alpha^{(k)}}\right)^{\Tilde{R}^{0n}}.
\end{equation*}
 From Markov's inequality it follows
\begin{equation*}
    \pr(M_n>\epsilon_n) \le \frac{1}{\epsilon_n} \prod_{k=1}^K \left( \frac{\alpha^{(k)}}{1+\alpha^{(k)}}\right)^{\Tilde{R}^{0n}} = \prod_{k=1}^K \left( \frac{1}{\epsilon_n^{1/\Tilde{R}^{0n}}}\frac{\alpha^{(k)}}{1+\alpha^{(k)}}\right)^{\Tilde{R}^{0n}}.
\end{equation*}
Assuming $\epsilon_n^{1/\Tilde{R}^{0n}} > \alpha^{(k)}/(1+\alpha^{(k)})$ for any $k=1\dots,K$ we conclude that
\begin{equation*}
\pr(M_n>\epsilon_n) \le \prod_{k=1}^K\exp(-C \Tilde{R}^{0n}) = \exp (-CK\Tilde{R}^{0n}),
\end{equation*}
for a positive constant $C>1$. 

\end{proof}

\begin{proof}[Proof of Theorem 2]
We denote with $f(\cdot)$ the density of $y_i$ under model (9), while $f_0(\cdot)$ denotes the density of $y_i$ under the true model (10) in the main paper.
Let $\Pi(\cdot)$ be the prior measure on $\eta_i = \prod_{k=1}^K u_i^{(k)}$.
Let $E_n$ be the target event $E_n = \{M_n > \epsilon_n\}.$ Consistent with \cite{castillo2012}, we write 
\begin{equation} \label{eq:num_den}
    \pr(E_n \mid y) = \frac{\int_{E_n} \prod_{i=1}^n \frac{f(y_i)}{f_0(y_i)} d\Pi(\eta_i)}{\int \prod_{i=1}^n \frac{f(y_i)}{f_0(y_i)} d\Pi(\eta_i)}= \frac{N_n}{D_n}.
\end{equation}
We introduce an event $A_n$ with large probability under the true data generating process
\begin{equation*}
    A_n = \{D_n \ge e^{-r^2_n}\pr(||\eta_H|| \le r_n)\},
\end{equation*}
where $H=\{l=\Tilde{R}^{0n}+1, \dots,n\}$.
If we decompose the probability in \eqref{eq:num_den} in the sum of two complementary conditional probabilities, it follows from the law of total probability that 
\begin{equation*}
    E_0\{\pr(E_n \mid y_n)\} \le E_0\{\pr(E_n \mid y_n) \mathds{1}(A_n)\}+\pr_0(A_n^c).
\end{equation*}
Thanks to Lemma 5.2 of \cite{castillo2012} $\pr_0(A_n^c) \le \exp(-r_n^2)$ for increasing $n$, from which follows that using \eqref{eq:num_den} we obtain
\begin{equation*}
    E_0\{\pr(E_n \mid y_n)\} \le \frac{\pr(E_n)}{e^{-r^2_n}\pr(||\eta_H|| \le r_n)}+ \exp(-r_n^2).
\end{equation*}
We can use Lemma 2, reported below, to bound the denominator. Given that $\pr(||\eta_H|| \le r_n) \ge \pr(|\eta_{\Tilde{R}^{0n}}| \le r_n/    \sqrt{n})^{n-\Tilde{R}^{0n}}$, it follows by Lemma 2 that
\begin{equation*}
    \pr(||\eta_H|| \le r_n) \ge \left[\prod_{k=1}^K \left\{ 1- \left(\frac{\alpha^{(k)}}{1+\alpha^{(k)}}\right)^{\Tilde{R}^{0n}} \right\}\right]^{n-\Tilde{R}^{0n}}.
\end{equation*}
For high $n$ and a constant $A>1/2$ we have $\{\alpha^{(k)}/(1+\alpha^{(k)})\}^{\Tilde{R}^{0n}}<A \Tilde{R}^{0n}/n$ for each $k=1,\dots,K$. Hence, we can write
\begin{equation*}
     \pr(||\eta_H|| \le r_n) \ge \prod_{k=1}^K \left(1- \frac{A\Tilde{R}^{0n}}{n} \right)^n = \left(1- \frac{A\Tilde{R}^{0n}}{n} \right)^{Kn}.  
\end{equation*}
Using $(1-x)^{1/x} >1/(2e)$ for $0<x<0.5$ it follows that $\pr(||\eta_H|| \le r_n)  \ge (1/2e)^{KA\Tilde{R}^{0n}}$.
Then, choosing $r_n^2 = \Tilde{R}^{0n}$ it follows by Lemma 1 that
\begin{equation} \label{eq:proofE0}
      E_0\{\pr(E_n \mid y_n)\} \le \exp[-\Tilde{R}^{0n} \{CK-1-KA\log(2e) \} ]+ \exp(-\Tilde{R}^{0n}).
\end{equation}
Since $C>2Ae$ implies $C>1+KA\log(2e)$ for $A>1/2$ and $K > 1$, \eqref{eq:proofE0} yields $E_0\{\pr(E_n\mid y_1,   \dots, y_n)\} \to 0$.

\end{proof}

\begin{lemma}
    Under the multiway CUSP prior in (4) in the main paper with  $\theta_\infty = \theta_{\infty}(n) < (r/\sqrt{n})^{1/K}$, for $1<r<\sqrt{n}$ ,
    \begin{equation*}
        \pr\left(|\eta_h| \le \frac{r}{\sqrt{n}}  \right) >  \prod_{k=1}^K \left\{ 1- \left(\frac{\alpha^{(k)}}{1+\alpha^{(k)}}\right)^h \right\}.
    \end{equation*}
\end{lemma}

\begin{proof}
We start by considering $|\eta_h| = |\prod_{k=1}^K u_h^{(k)}|$ and observing that $|\prod_{k=1}^K u_h^{(k)}| = \prod_{k=1}^K |u_h^{(k)}|$, the following inequalities hold

\begin{align*} 
\pr\left(|\eta_h| \le \frac{r}{\sqrt{n}} \mid \pi_h^{(1)},\dots, \pi_h^{(K)} \right) &\ge 
\pr\left\{ \left(\max_{k \in \{1,\dots,K\}} |u_h^{(k)}|\right)^{K} \le \frac{r}{\sqrt{n}} \mid \pi_h^{(1)},\dots, \pi_h^{(K)}
\right\} \\
&= \pr\left\{\max_{k \in \{1,\dots,K\}} |u_h^{(k)}| \le \left(\frac{r}{\sqrt{n}}\right)^{1/K} \mid \pi_h^{(1)},\dots, \pi_h^{(K)}
\right\} \\
&= \prod_{k=1}^K \pr\left\{|u_h^{(k)}| \le \left(\frac{r}{\sqrt{n}}\right)^{1/K} \mid \pi_h^{(1)},\dots, \pi_h^{(K)}
\right\}. 
\end{align*}
Considering each factor of the product of the last line of the equation above we have 
\begin{align*}
    \pr\left\{|u_h^{(k)}| \le \left(\frac{r}{\sqrt{n}} \right)^{1/K} \mid \pi_h^{(1)},\dots, \pi_h^{(K)} \right\} & =  (1- \pi_h^{(k)}) \pr\left\{|t_{2a_\theta;(a_\theta/b_\theta)^{1/2}}|\leq \left(\frac{r}{\sqrt{n}}\right)^{1/K}\right\} + \pi_h^{(k)}\\ & \ge\pi_h^{(k)}, 
\end{align*}
where $t_{\nu;s}$ is a Student-$t$ distributed random variables with scale $s$ and $\nu$ degrees of freedom. Hence 
\begin{align*}
    \pr\left(|\eta_h| \le \frac{r}{\sqrt{n}} \mid \pi_h^{(1)},\dots, \pi_h^{(K)} \right) &\ge \prod_{k=1}^K \pi_h^{(k)},
\end{align*}
and, integrating out $\pi_h^{(k)}$ for $k=1,\dots,K$ we obtain
\begin{align*}
     \pr\left(|\eta_h| \le \frac{r}{\sqrt{n}}\right) \ge \prod_{k=1}^K   E(\pi_h^{(k)}) = \prod_{k=1}^K \left\{ 1- \left(\frac{\alpha^{(k)}}{1+\alpha^{(k)}}\right)^h \right\}.
\end{align*}
\end{proof}

\section{Additional details on posterior inference} 

In this section, we provide practical details on the implementation of the adaptive Gibbs sampler outlined in Algorithm 2 in the main paper and discuss posterior inference for the probit-AdaTuck model.
To let the chain stabilize, we start the adaptation after a fixed number of iterations $\Tilde{t}$. In our analysis, we set $\Tilde{t} = 400$. The adaptation is performed with probability $p(t) = \exp(c_0+c_1t)$ and following \cite{Bhattacharyadunson11} we set $(c_0, c_1)=(-1, -5\times10^{-4})$. It is reasonable to initialize $(R_1, \dots, R_K)$ with sufficiently large values such that $R_k \le n_k+1$. We found in practice that initializing $R_k$ at $\{\mathrm{max} (n_1, \dots, n_k) + n_k\}/3$ for $k=1,\dots,K$ yields  a good trade-off in terms of conservative truncation levels and computational scalability. Lastly, we initialize $R^{*}_k$ to $R_k-1$. 

Posterior inference for the probit-AdaTuck model defined in (7) of the main paper proceeds through an adaptive Gibbs sampler which exploits a data-augmentation scheme for the latent tensor $\Q$ \citep{albert1993bayesian}. 
The adaptation strategy is the same as discussed in Section 3 of the main paper, so posterior inference still follows Algorithm 2 in the main paper. However, what differs for the probit-AdaTuck model are the Gibbs sampler's sampling steps. Algorithm \ref{algo:gibbsBin} outlines these specific steps for the probit-AdaTuck model.
For dealing with missing entries we follow the same strategy discussed in the main paper, i.e we run the sampler conditioned only on the observed data. Specifically, we sample $q_{\mathbf{i}}$ from a univariate truncated normal (step 1 in Algorithm \ref{algo:gibbsBin}) only for the corresponding observed $y_{\mathbf{i}}$, resulting in a computationally efficient procedure.

\begin{algorithm}[h]
    \caption{One cycle of the Gibbs sampler for the probit-AdaTuck model with finite truncation level $(R_1, \dots, R_K)$.} \label{algo:gibbsBin}
    \begin{enumerate}
    \item   For $\mathbf{i} \in I$:\\
    if $y_{\mathbf{i}} =1$ sample  $q_{\mathbf{i}}$ from $TN_{[0, \infty)}(z_{\mathbf{i}}, 1)$, where $TN_{[a, b)}(m,v)$ is a normal density of mean $m$ and variance $v$ truncated to the region $[a, b)$ and $z_{\mathbf{i}} = \sum_{k=1}^{K} \sum_{r_k=1}^{R_k}  g_{\mathbf{r}} \uu[1]{r_1} \circ \cdots  \circ \uu[K]{r_K}$;
    otherwise sample  $q_{\mathbf{i}}$ from $TN_{(-\infty, 0)}(z_{\mathbf{i}}, 1)$. 
    
    \item  For $k$ from 1 to $K$ and for $i_k$ from 1 to $n_k$:\\
     $\quad$ sample the $i$-th row of $\UU{k}$ from   $N_{R_k}(V_u G_{(k)} W^{(-k)} q_{i_k}, \ V_u)$, with $q_{i_k}$ the $i_k$-th\\ $\quad$ row of $Q_{(k)}$ and
    $V_u = \{\tdiag(\theta_{1}^{(k)}, \dots, \theta_{R_k}^{(k)})^{-1} + G_{(k)} W^{(-k)}   W^{(-k) \top} G_{(k)}^{\top}\}^{-1}$. 

    \item Sample $\tvec(\G)$ from $N_R(V_g W^{\top} \tvec(\Q), V_g )$,  with $V_g = (\Psi^{-1} +\sigma^{-2} W^{\top}  W)^{-1}$. 
    
    \item  Follows steps (4)--(9) of Algorithm 1in the main paper.

    \end{enumerate}

\end{algorithm}

\section{Image completion}

In this section, we analyze image data to assess AdaTuck's performance in image reconstruction.
A Red-Green-Blue (RGB) image can be naturally encoded in a three-way tensor, with  dimension pixel height $\times$ pixel width $\times$ 3, where each entry corresponds to the intensity of the red, green, or blue color channels.
The dataset used in our analysis is a cat image, depicted in panel (a) of Figure \ref{fig:RH}, with dimensions $120 \times 107 \times 3$. To assess the performance of AdaTuck in image reconstruction we randomly impute to missing values the $50\%$ and the $70\%$ of the total entries of the data $\Y$.

\begin{figure}
\begin{center}
    \includegraphics[width=0.9\textwidth]
    {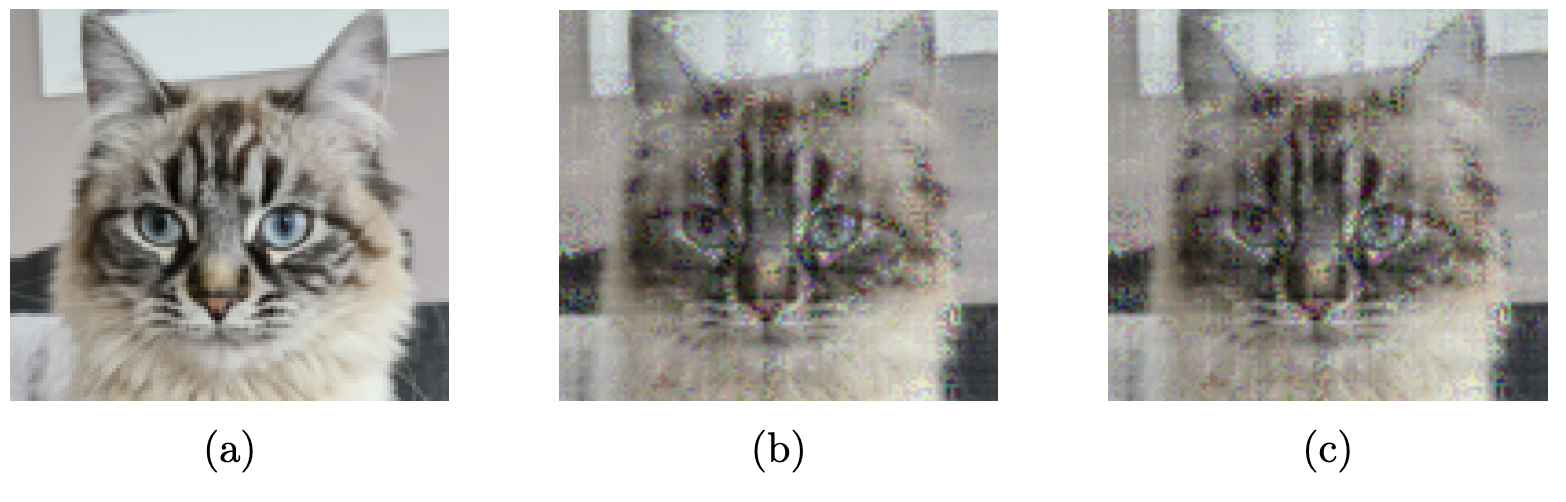} 
    \caption{Original image (a); image reconstructed by AdaTuck with $50\%$ missing values (b) and image reconstructed by AdaTuck with $70\%$ missing values (c).}
    \label{fig:RH}
\end{center}
\end{figure}

We apply AdaTuck to the image data, using the same hyperparameters and number of iterations for the Gibbs sampler described in Section 4.1 of the main paper. 
We benchmark the results with the  Bayesian tensor factorization method proposed in \cite{khan2016tensorbf} (\texttt{tensorBF}).
We compare the two methods by computing the  Monte Carlo MSE in predicting the missing values with the reconstructed tensor based on the latent structure. 
Table \ref{tab:Image} summarizes for the two scenarios the median and the interquartile range of the MSE obtained by AdaTuck and TensorBF.
In both scenarios, AdaTuck outperforms the competitor in predicting the missing entries of the tensor. 
Figure \ref{fig:RH} shows the cat image reconstructed by AdaTuck with $50\%$ missing values (panel b) and $70\%$ missing values (panel c).

\begin{table}[h]
    \centering
    \caption{Performance of AdaTuck and TensorBf in predicting the missing entries for the image data. The shorthand `NA(\%)' represents the percentage of missing entries with respect to the total tensor entries.}
    \vspace{0.3cm}
    \begin{tabular}{c|cc|cc}
    
   \multicolumn{1}{}{} &
      \multicolumn{2}{c}{AdaTuck} &
      \multicolumn{2}{c}{ TensorBF}  \\
      
    NA (\%) & MSE  &  MSE  & MSE  &  MSE    \\
    & (median) & (IQR) & (median) & (IQR) \\
      \hline 
   50  & 0.012 & 0.003 & 0.014 & 0.002 \\
     70 & 0.013  & 0.002 & 0.016 & 0.004 \\
    
    \hline 
  \end{tabular}
    \label{tab:Image}
\end{table}

\section{Additional results for illustrations}

\begin{table}[h]
    \centering
    \caption{Median and IQR of the latent multi-ranks estimated through the eigenvalue ratio (ER) and the Growth Ratio (GR) for the 20 different simulations in the two scenarios for $30\%$ and $50\%$ of total entries missing.}
    \vspace{0.3cm}
    \begin{tabular}{cc|ccc|ccc}
    
   \multicolumn{2}{}{} &
      \multicolumn{3}{c}{ER} &
      \multicolumn{3}{c}{GR}  \\
      
    $(R_1^0, R_2^0, R_3^0)$ & NA (\%) & $R_1$  &  $R_2$ & $R_3$ &$R_1$  &  $R_2$   &$R_3$ \\
    &&& median(IQR) &&& median(IQR)&\\
    \hline
   (5,5,5) & 30  & 1(1)  & 1(1) & 1(0.25)& 1.5(2)& 1(1) & 1(1)\\    
   & 50 & 1(0)  & 1(0) & 1(1)& 1(0.25)& 1(0) & 1(0.25)\\  
    (10,7,3) & 30  & 1(1)  & 1(1) & 1(0)& 1(1)& 1.5(1) & 1(0)\\
    & 50  & 1(1)  & 1(1) & 1(0.25) & 1(1)& 1(1) & 1(0)\\
    \hline 
  \end{tabular}
    \label{tab:rankDet}
\end{table}

\begin{figure}[h]
\begin{center}
    \includegraphics[width=0.99\textwidth]{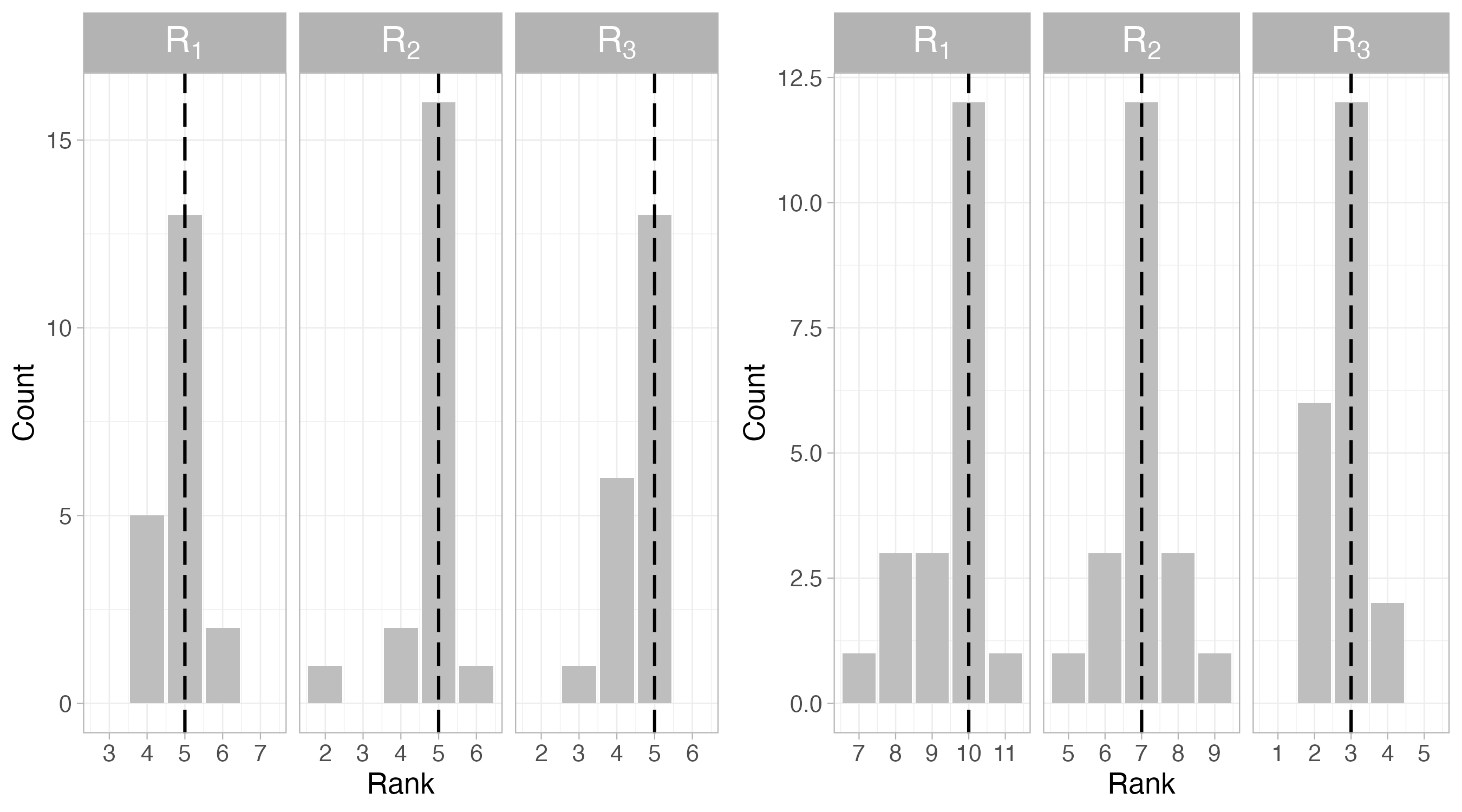} 
    \caption{Distribution of the Monte Carlo estimations of probit-AdaTuck for the median of $R_1^*$, $R_2^*$ and $R_3^*$ for the 20 different simulations in the two scenarios for $30\%$ of total entries missing. Black dotted lines are the true $R_1^0$, $R_2^0$ and $R_3^0$ values.}
    \label{fig:rankBin}
\end{center}
\end{figure}

\begin{figure}[h!]
\begin{center}
    \includegraphics[width=0.99\textwidth]{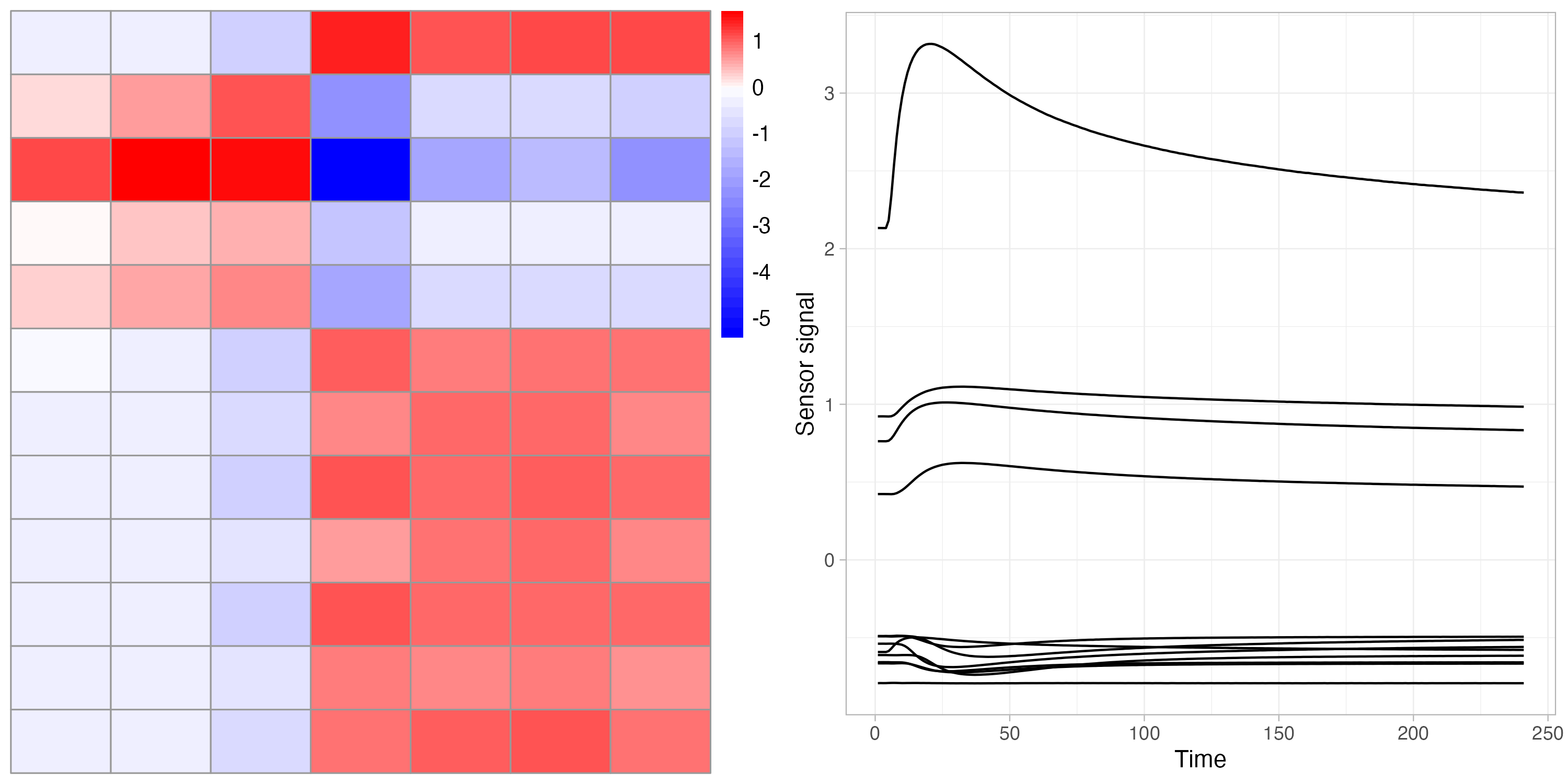} 
    \caption{Posterior sample of $U^{(3)}$ (left panel) and signal posterior means for the twelve sensors over time for one licorice sample (right panel).}
    \label{fig:postSensor}
\end{center}
\end{figure}

\begin{figure}[h!]
\begin{center}
\includegraphics[width=0.99\textwidth]{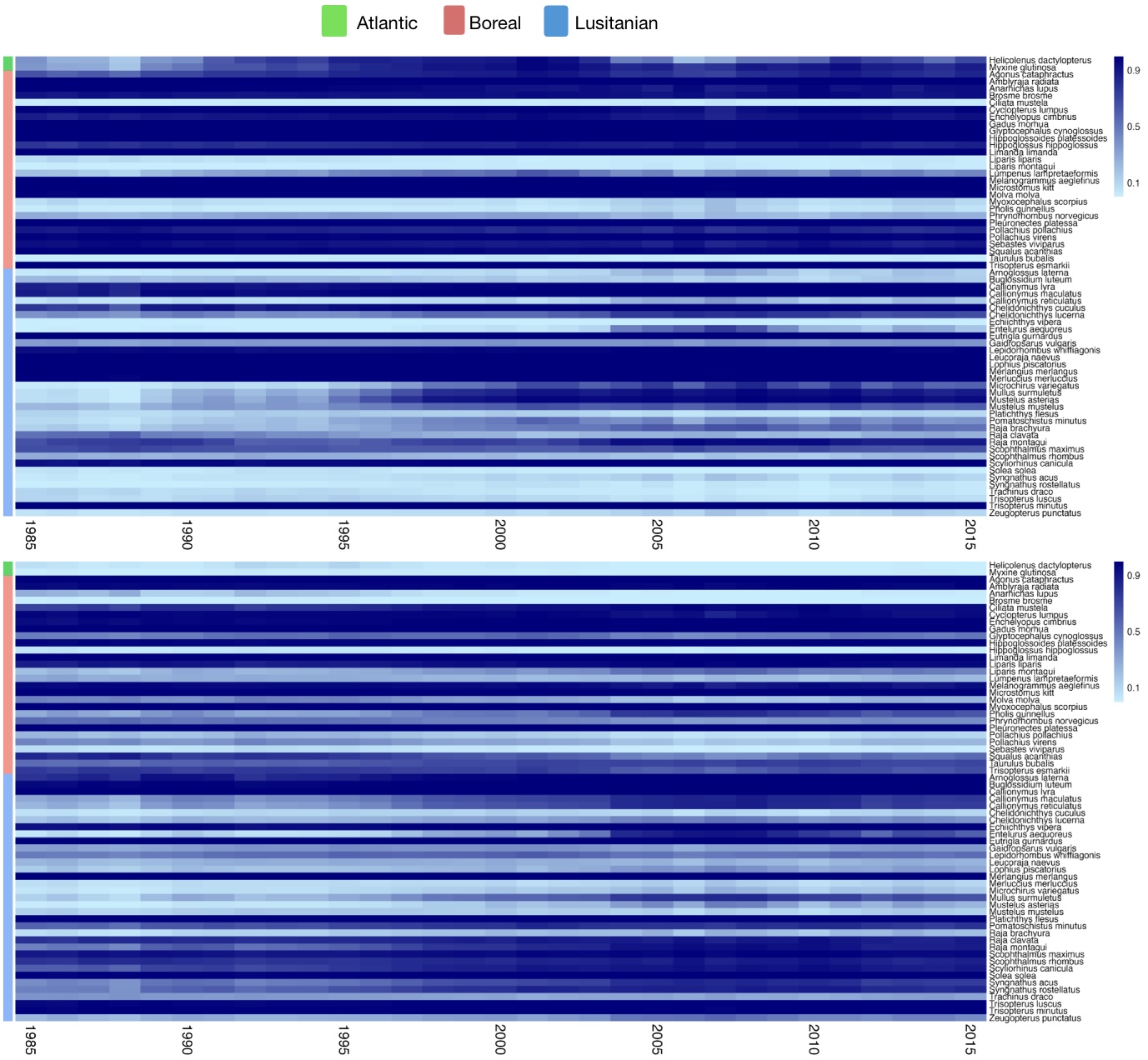} 
    \caption{Posterior mean of the frontal slices of $\mathcal{M}$ for the north area (top panel) and the southeast area (bottom panel). }
    \label{fig:Siteheat}
\end{center}
\end{figure}

\begin{figure}[h!]
\begin{center}
\includegraphics[width=0.85\textwidth]{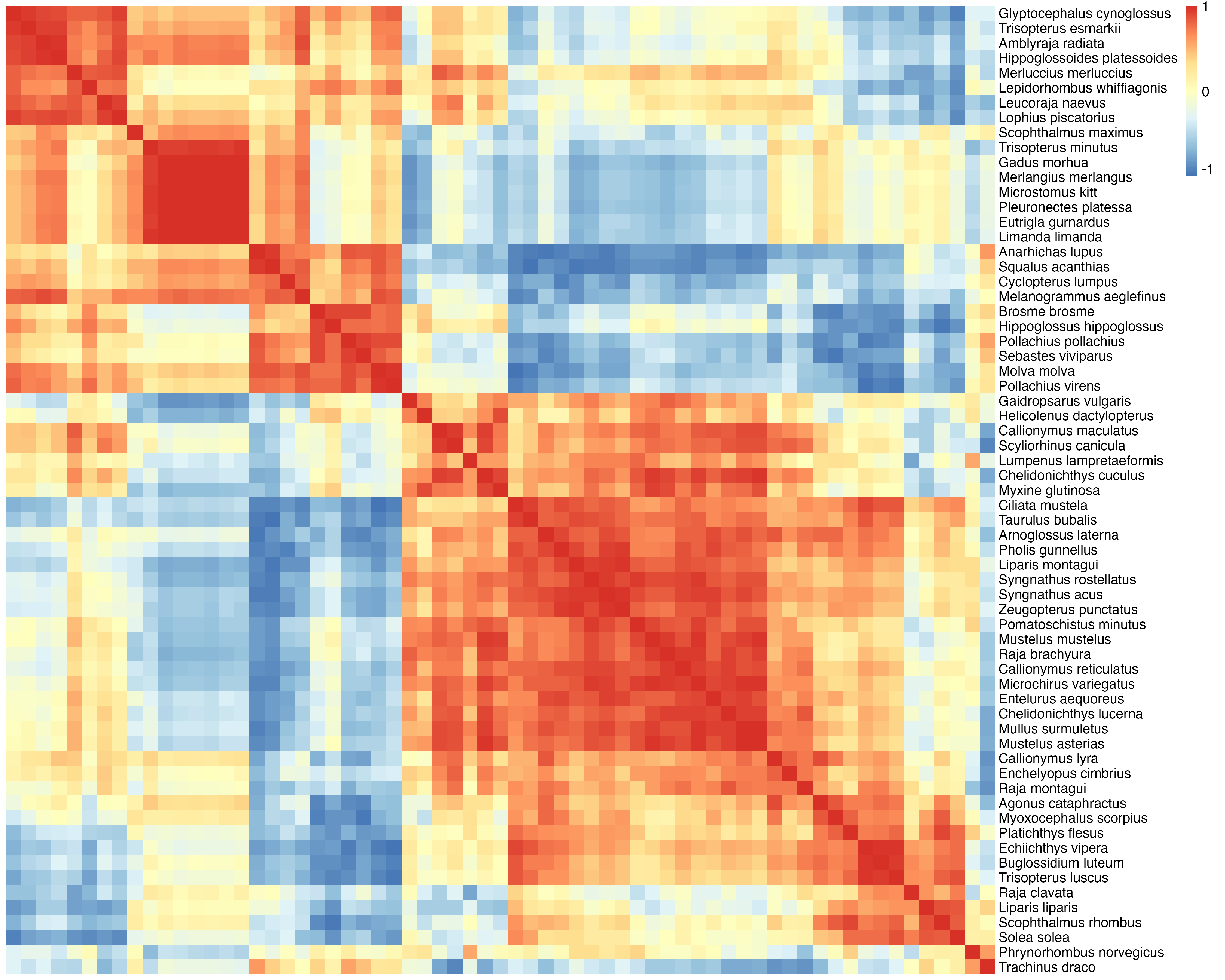} 
    \caption{Posterior mean of species co-occurrence correlation matrix for the first year in the  northernmost area.}
    \label{fig:corr}
\end{center}
\end{figure}

\end{document}